\documentclass[sigconf, nonacm]{acmart}

\newcommand\vldbdoi{XX.XX/XXX.XX}

\newcommand\vldbavailabilityurl{URL_TO_YOUR_ARTIFACTS}
\newcommand\vldbpagestyle{plain} 

\usepackage{xspace}

\usepackage{graphicx}
\usepackage{balance}  

\usepackage{float}

\usepackage{amsmath}
\usepackage{setspace}

\usepackage[linesnumbered ,titlenumbered,ruled, noend]{algorithm2e}

\usepackage{amsthm}

\usepackage{eucal}
\usepackage{amsfonts}

\theoremstyle{definition}

\DeclareMathOperator*{\argmax}{arg\,max}

\DeclareMathDelimiter{(}{\mathopen} {operators}{"28}{largesymbols}{"00}
\DeclareMathDelimiter{)}{\mathclose}{operators}{"29}{largesymbols}{"01}

\newtheoremstyle{theoremdd}
  {\topsep}
  {\topsep}
  {\itshape}
  {0pt}
  {\bfseries}
  {. ---}
  { }
  {\thmname{#1}\thmnumber{ #2}\textnormal{\thmnote{ (#3)}}}

\newtheorem{theorem}{Theorem}

\theoremstyle{theoremdd}
 
\usepackage{fancyvrb}
\usepackage{enumitem}

\usepackage{algorithmic}
\usepackage{textcomp}
\usepackage{xcolor}
\def\BibTeX{{\rm B\kern-.05em{\sc i\kern-.025em b}\kern-.08em
    T\kern-.1667em\lower.7ex\hbox{E}\kern-.125emX}}

\newtheorem{defn}{Definition}[section]

\newtheorem*{example*}{Example}

\newcommand{\parcur}{ParCuR}
\newcommand{\roulette}{RouLette}

\usepackage{tikz}

\usepackage{mathtools}

\begin{document}
\title{Real-Time Analytics by Coordinating Reuse and Work Sharing}

\author{Panagiotis Sioulas$^0$}
\affiliation{%
  \institution{Oracle}
  \streetaddress{}
  \city{Zürich}
  \state{Switzerland}
}
\email{panagiotis.sioulas@oracle.com}

\author{Ioannis Mytilinis$^0$}
\affiliation{%
  \institution{Oracle}
  \streetaddress{}
  \city{Zürich}
  \state{Switzerland}
}
\email{ioannis.mytilinis@oracle.com}

\author{Anastasia Ailamaki}
\affiliation{%
  \institution{EPFL}
  \streetaddress{}
  \city{Lausanne}
  \state{Switzerland}
  \postcode{1015}
}
\email{anastasia.ailamaki@epfl.ch}

\begin{abstract}
Analytical tools often require real-time responses for highly concurrent parameterized workloads. A common solution is to answer queries using materialized subexpressions, hence reducing processing at runtime. However, as queries are still processed individually, concurrent outstanding computations accumulate and increase response times. By contrast, shared execution mitigates the effect of concurrency and improves scalability by exploiting overlapping work between queries but does so using heavyweight shared operators that result in high response times. Thus, on their own, both reuse and work sharing fail to provide real-time responses for large batches. Furthermore, naively combining the two approaches is ineffective and can deteriorate performance due to increased filtering costs, reduced marginal benefits, and lower reusability.

In this work, we present \parcur, a framework that harmonizes reuse with work sharing. \parcur\ adapts reuse to work sharing in four aspects: i) to reduce filtering costs, it builds access methods on materialized results, ii) to resolve the conflict between benefits from work sharing and materialization, it introduces a sharing-aware materialization policy, iii) to incorporate reuse into sharing-aware optimization, it introduces a two-phase optimization strategy, and iv) to improve reusability and to avoid performance cliffs when queries are partially covered, especially during workload shifts, it combines partial reuse with data clustering based on historical batches. \parcur\ outperforms a state-of-the-art work-sharing database by $6.4\times$ and $2\times$ in the SSB and TPC-H benchmarks respectively.
\end{abstract}

\maketitle

\pagestyle{\vldbpagestyle}

\footnotetext{This work was done while the author was at EPFL.}
\section{Introduction}

Reusability is a driving factor for many analytical tools, such as dashboards, notebooks, and pipelines. Often, such reusable workloads consist of highly concurrent parameterized queries. Dashboards, for example, produce visualizations by processing several canned queries that are parameterized through UI interactions or other queries \cite{tableau, idebench}. Similarly, analysts rerun data-science notebooks for reproducibility and exploration, often with different parameters  \cite{storynotebook, explorationandexplanation, jupyterrepro, appyters}; hence, multiple queries that transform and analyze data recur. While such applications process large numbers of queries, they are interactive in nature and require low response times for all queries. However, under high concurrency, backend databases struggle to produce responses within a tight timeframe.

Traditionally, there are two approaches to accelerate processing for large batches of recurring queries. On the one hand, we can optimize individual queries. To do so, both commercial and open-source databases can \textit{reuse} materialized results; databases avoid full recomputation and drastically reduce processing time. Often, optimizing for reuse opportunities is automated in the form of caching, recycling, and materialized views and subexpressions \cite{shimcaching, ivanovarecycle, roussopoulosviews, zhou2007efficient, jindal2018selecting}. Nevertheless, materialization is subject to a storage budget and thus leaves outstanding computations. Moreover, as the outstanding computations for different queries are still processed individually, response time is increased with concurrency. 

On the other hand, we can optimize the scalability of batch processing using \textit{work sharing}. Work-sharing databases reduce the total processing time by exploiting overlapping computations across the queries in the batch. However, large numbers of heavyweight shared operators and the fact that everything is recomputed from scratch can violate stringent response time requirements.

Figure \ref{parcur:fig:trends} depicts processing time for a large query batch\footnote{The setup corresponds to Figure \ref{parcur:fig:budget-ab}a with 50\% budget (presented in Section \ref{parcur:sec:caching-exp}).}. Both query-at-a-time (QaT) reuse and work sharing fail to provide fast responses. Reuse eliminates computations by precomputing joins, but suffers from concurrent outstanding processing (i.e., filters on materialized results, non-materialized joins). By contrast, work sharing mitigates the impact of concurrency and reduces the response time but suffers from processing heavy shared joins at runtime.

\begin{figure}[t]
 \centering
	\includegraphics[width=0.8\hsize, page=18, trim=0 0 0 200, clip]{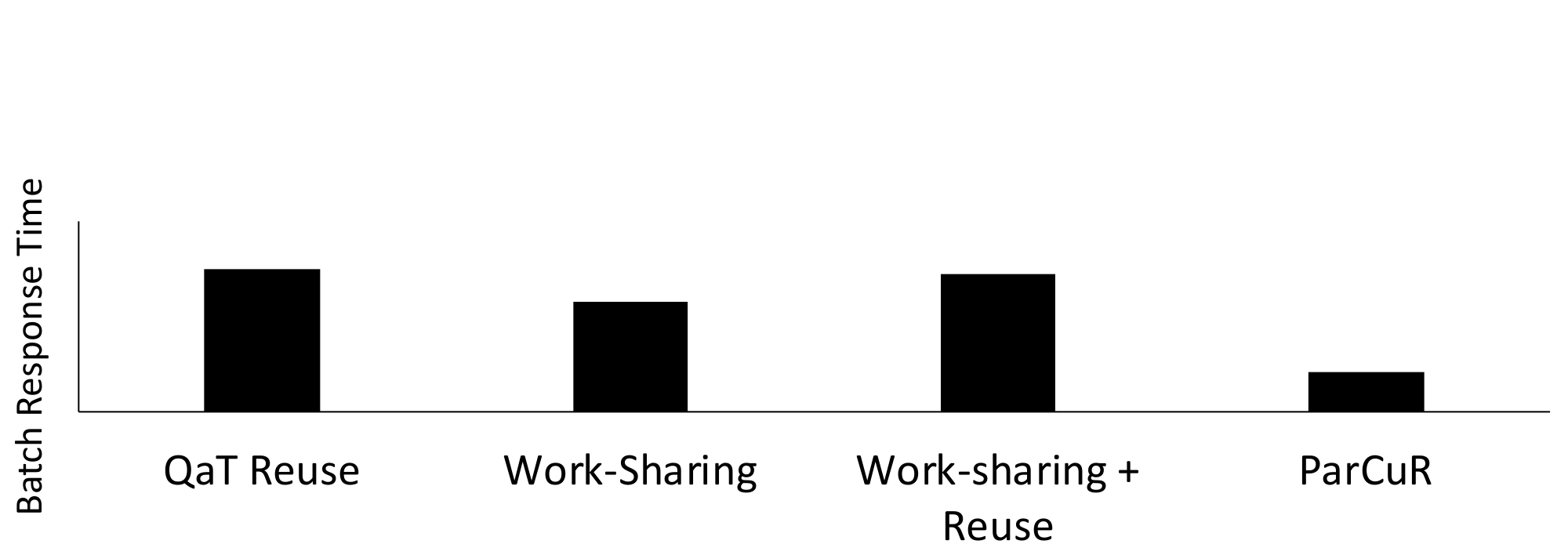}
 \vspace{-3mm}
    \caption{\parcur\ harmonizes reuse and work sharing to speed up recurring batches }
	\label{parcur:fig:trends}
    \vspace{-3mm}
\end{figure}

Individually, both reuse and work sharing fail to process large workloads interactively but still make complementary contributions. Thus, it is attractive to combine the two approaches to exploit their cumulative benefit. However, naively reusing materialized results in a work-sharing database as we would in a query-at-a-time database brings limited benefit and can even degrade performance ("Work-sharing + Reuse" in Figure \ref{parcur:fig:trends}). Reuse in a work-sharing environment is ineffective because i) it eliminates upstream shared operators only when their results are not required by any downstream computation,
ii) as it rewrites only queries that the used materialized results subsume, mismatching (i.e., non-subsumed) queries may recompute, fully or partially, the reused results, hence decreasing benefit -- mismatches become increasingly likely as concurrency is increased, especially during workload shifts -- and
iii) it severely amplifies processing for shared filters.

To enable interactive responses for large parameterized batches, we introduce \parcur\ (Partition-Cut-Reuse), a novel framework that harmonizes reuse with work sharing. To address the limited effectiveness of reuse in work-sharing environments, \parcur\ adapts materialization and reuse techniques across three axes:

\emph{Cut:} Work sharing violates the assumptions of traditional subexpression selection \cite{zhou2007efficient, jindal2018selecting}; thus, existing solutions fail to minimize processing. To increase the impact of reuse, \parcur\ introduces novel materialization and reuse policies that make decisions based on the eliminated shared operators in the work-sharing setup. As eliminating each shared operator depends on downstream decisions, \parcur\ introduces the concept of \emph{cuts}. Cuts represent sets of materialized subexpressions that act synergistically in eliminating more upstream operators. The policies use cuts when evaluating which results to materialize or reuse. \parcur\ proposes approximation algorithms for materialization as well as a cost-based reuse algorithm that maximizes processing-time savings.

\emph{Reuse:} \parcur\ focuses on making reuse efficient, and thus it is imperative to reduce the high processing time for shared filters. To this end, it builds and uses access methods on materialized subexpressions. By building access methods on materialized subexpressions based on frequent predicates, and by using the access methods at runtime, \parcur\ evaluates frequent filters for one batch of tuples at a time, thus amortizing the required processing.

\emph{Partition:} To increase the usability of materialized results in case of mismatches, e.g., during workload shifts, \parcur\ uses partial reuse. It uses the fragments of materialized results that are relevant for each query batch at hand and performs any additional recomputation only as needed, thus relaxing the subsumption constraint; it eliminates shared operators for all queries for the data ranges that materialized results cover. To efficiently identify and access the relevant fragments and the base data for the recomputation, \parcur\ uses partitioning. Nevertheless, by materializing and reusing at the partition-granularity, it creates a dependency between the storage footprint and the partitioning scheme: such materializations may include tuples that are rarely useful if the partition is misaligned with the predicates of the corresponding queries. Hence, to maximize reuse while minimizing footprint, \parcur\ introduces a novel partitioning algorithm that clusters together data that are accessed by similar subexpressions and hence aligns partitions with predicate-subexpression combinations.

\parcur\ incorporates the above techniques in a two-phase framework: i) an offline tuner that optimizes \parcur's state (i.e., partitions, materialized results, access methods) for a target workload and ii) an online executor that, by exploiting the available state, minimizes the processing time for query batches arriving at runtime. 
By adapting and exploiting the available state, \parcur\ makes reuse efficient and effective in work-sharing environments. As Figure \ref{parcur:fig:trends} demonstrates, \parcur\ drastically reduces batch response time. The experiments show that \parcur\ outperforms work sharing by $6.4\times$ and $2\times$ in the SSBM and TPC-H benchmarks, respectively.

We make the following contributions:

\vspace{-3mm}
\begin{itemize}
\item
Choosing materializations using QaT heuristics is ineffective and uses the storage budget suboptimally. We propose a family of materialization policies that, by adapting to the workload's sharing opportunities, improve time savings for the same budget. 
\item
Work-sharing decisions and access patterns affect the benefit of reuse. We propose a cost-based optimization strategy that chooses when and which materializations to inject into each batch's plan such that response time is minimized. 
\item
Naively reusing materializations in work-sharing databases can increase response time considerably. Instead, we propose that materialization should be accompanied by access methods that enable data skipping and filter skipping.
\item
Increasing the usability of materialized results in case of mismatches requires partial reuse. Partition-level materialization and execution enable efficient partial reuse at the expense of storage overhead. We propose a novel partitioning scheme that maximizes reuse while minimizing redundant materialization by aligning partitions to workload patterns.
\end{itemize}

\section{Reuse in Shared Execution}

We provide an overview of the challenges in reusing materializations during shared execution. We first briefly present work-sharing concepts and motivate reusing materializations to reduce recomputation, then highlight the performance pitfalls that reuse introduces when combined with work sharing, and finally outline our solutions. For ease of presentation, we use the following batch as a running example:

\begin{verbatim}
Q1: SELECT SUM(X) FROM A,B,C,D WHERE expr1
Q2: SELECT SUM(X) FROM A,B,E WHERE expr2
\end{verbatim}

\subsection{Shared Execution}

Work-sharing databases accelerate query batches by exploiting overlapping work across queries. To do so, they rely on i) the \emph{global plan} and ii) the \emph{Data-Query model}.



\textbf{Global plan:} The global plan expresses sharing opportunities among different queries. It is a directed acyclic graph (DAG) of relational operators that process tuples for one or more queries, and multi-cast their results to one or more parent operators. Figure \ref{parcur:fig:motivation-new} shows the global plan for Q1 and Q2. For ease of reference, each operator is labeled with a number. Operator $2$ processes $A \bowtie B$ for both queries, and sends results to operators $3$ and $8$, which serve Q1 and Q2, respectively. At the two roots, the global plan produces the results for Q1 and Q2. By processing each operator of the global plan only once, the database shares work across queries and reduces the overall processing time.

\textbf{Data-Query model:} The Data-Query model enables efficient sharing between queries with different selection predicates. Sharing through query re-writing that uses standard relational operators and filters the union of the predicates is expensive as i) it produces and processes redundant tuples, and ii) filters data several times within the plan \cite{precision-sharing}. For example, operator $2$ can join a ``probe'' tuple belonging only to Q1 with a ``build'' tuple only belonging to Q2 and filter it out afterwards. The Data-Query model addresses these two inefficiencies: it annotates each tuple with a \textit{query-set} that indicates to which queries the tuple contributes. Then, specialized \emph{shared} operators process both the actual tuples and the query-sets. This way the database tracks membership for intermediate results, and eliminates redundant tuples early. Therefore, with Data-Query model: i) the global plan shares work on tuples that are common across some but not all queries, and ii) the operators can immediately drop tuples that do not belong to any query.

\subsection{Recomputation Bottleneck}
\label{parcur:motivation:bottleneck}

When using work sharing, processing more queries increases the response time sublinearly, and thus, the total processing time is reduced compared to QaT execution. However, for each submitted query batch, global plan execution always starts from a clean slate. Data flows from the input tables to each query's output, and all shared operators of the global plan are fully processed from scratch.

Recomputation of previously ``seen'' expressions can be critical as the additional processing for handling query-sets renders shared operators particularly time-consuming. For example, shared filters and joins require one or more query-set intersections, the cost of which is increased as a function of the number of queries. Furthermore, shared filters are not simple comparisons but are implemented as joins with predicates using the predicate indices. 

All in all, as global plans often consist of tens of operators, processing accumulates and prevents providing results within a tight time window. Therefore, to offer interactivity, we need to reduce the required computations for each batch.

\subsection{Pitfalls of Combining Reuse and Work Sharing}
\label{parcur:motivation:challenges}

Analytical databases reduce runtime computations by reusing precomputed results. However, we observe that using materializations in work-sharing environments exhibits a set of properties that have not been studied before and which make reuse inefficient. Namely, these properties are: i) \emph{shared cost}, ii) \emph{synergy}, iii) \emph{filter amplification}, and iv) \emph{risk of miss}. We elaborate on each of these properties.

\begin{figure}[t]
 \centering
	\includegraphics[width=0.4\hsize, page=1, trim=200 0 200 0]{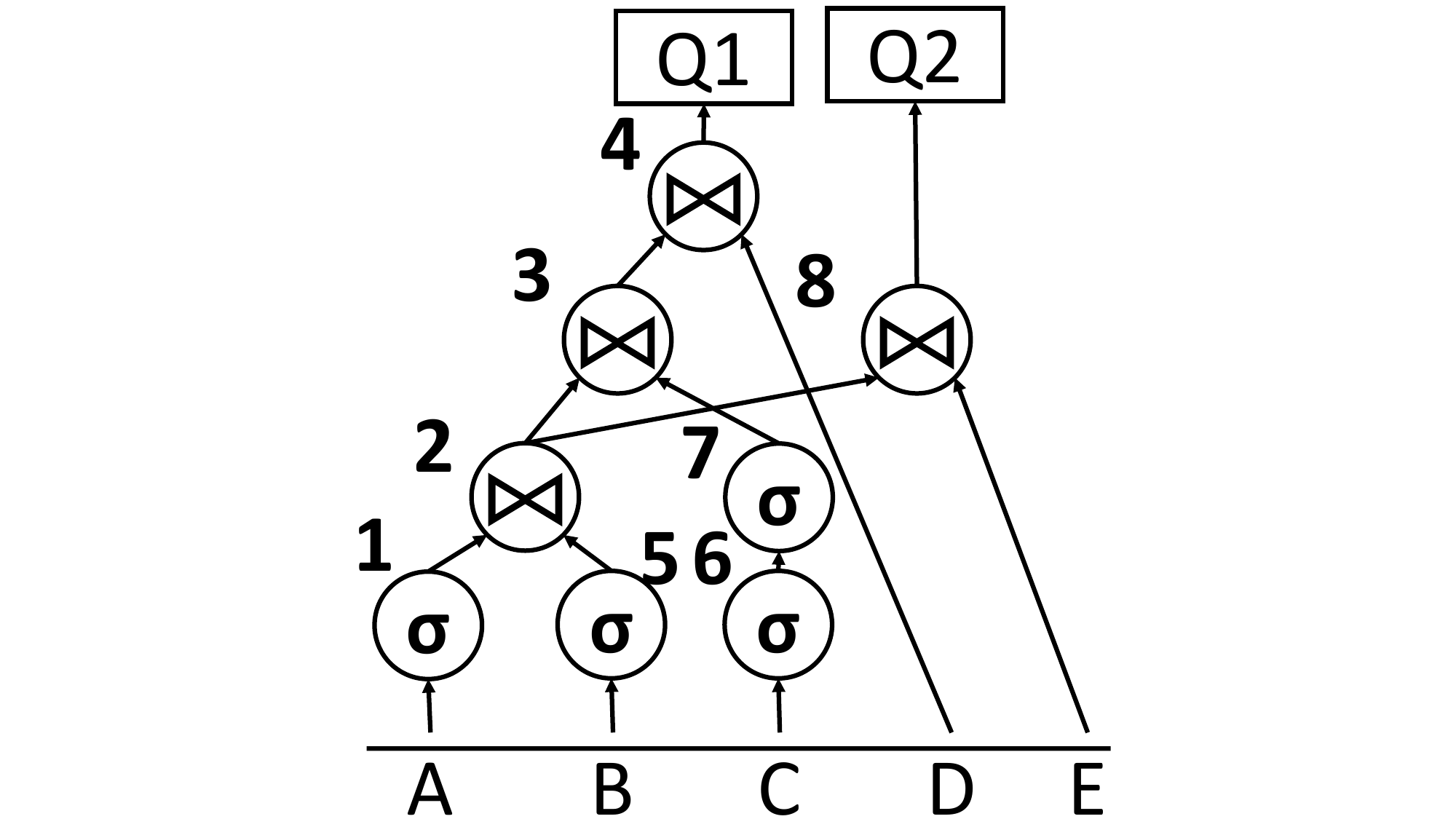}
    \vspace{-5mm}
    \caption{Motivational example: work sharing introduces novel challenges for reuse}
    \vspace{-5mm}
	\label{parcur:fig:motivation-new}
\end{figure}

\noindent \textbf{Shared cost:} \emph{QaT cost models are inaccurate in work-sharing environments}. Work sharing affects both which operators reuse eliminates and their relative importance. On the one hand, reuse eliminates upstream operators only as long as their results are not required by other remaining downstream operators. For example, reusing the results of operator $4$ eliminates operators $3$ and $4$, but operator $2$ is still required for Q2. On the other hand, work sharing across queries diminishes the importance of frequency of occurrence for operators; the savings depend more on the total number of Data-Query tuples processed by the shared operator rather than the number of participating queries. This is contrary to the assumptions of traditional cost models for materializing intermediates, which assume that reuse eliminates all upstream costs and which simply add up the benefit for each affected query. 

\noindent \textbf{Synergy:} \emph{The benefit of individual materializations is amplified.} Materialization decisions affect each other's results differently than they do in single-query plans. Reuse in single-query plans results in diminishing returns. For example, reusing the results of operator $4$ in the original plan eliminates joins $3$ and $4$, whereas reusing the same results in a rewritten plan that already uses operator $3$ only eliminates join $4$. This observation is critical for the design of heuristic materialization algorithms that are based on submodularity. However, diminishing returns are not necessarily the case in global plans. Consider the example where we reuse results for operators $3$ and $8$. We observe a counter-intuitive effect: individually, they eliminate one join each, but together the benefit is amplified, and they eliminate $3$ joins. This effect, which we refer to as \emph{synergy}, marks a departure from traditional materialization and reuse.

\noindent \textbf{Filter amplification:} \emph{Shared filters over materializations dominate the total processing time}. When injecting a materialization into a global plan, the work-sharing database needs to process filters from all the tables participating in the computation. For example, if the database reuses the results of the subquery corresponding to operator $4$, the global plan needs to process $4$ shared filters from tables $A$, $B$, and $C$. Then, the processing time for filters is amplified for two reasons: i) materializations can have a significantly larger cardinality than small dimension tables, and ii) filters must process every materialization where the corresponding table participates (e.g., filters from $A$ are processed on the materializations of both $4$ and $8$). In some cases, reuse deteriorates performance compared to processing the batch from scratch using work sharing.

\noindent \textbf{Risk of miss:} \emph{The probability that the materialization covers all accessed data decreases with the number of queries}. Reuse typically requires that the materialization fully subsumes the subquery that it eliminates. Similarly, the materialization needs to subsume all participating queries to eliminate subplans in global plans. For example, eliminating operator $2$ by reusing its result requires that both Q1 and Q2 can be answered using the materialization. Assume that the materialization only covers \textit{expr1} and \textit{expr1} defines a subset of \textit{expr2}: then, even if Q1 is answered using the materialized subexpression, Q2 fully recomputes the shared operator's result already and thus reuse brings no benefit compared to shared execution. Requiring full subsumption for materialized subexpressions has a high risk of mismatch, especially in case of workload shifts.

\subsection{Harmonizing Reuse and Work Sharing}

To significantly reduce their runtime computations, work-sharing databases need to address inefficiency in reuse. In this work, we harmonize work sharing and reuse: we redesign, based on the above-mentioned properties, the techniques for materializing and reusing precomputed results such that we maximize  eliminated computations and minimize reuse overhead. Harmonization takes place across three axes: i) materialization and reuse policies which address shared cost and synergy, ii) access methods for materializations which address filter amplification, and iii) partial reuse, which addresses the risk of miss.

\noindent \textbf{Materialization and reuse policies:} Due to shared cost and synergy, algorithms for selecting materializations or injecting materializations into plans make suboptimal decisions. Work sharing renders their cost models inaccurate and violates common submodularity assumptions. Hence, harmonization requires novel  materialization and reuse policies that, by taking into account both shared cost and synergy, select materializations that bring higher processing time reduction, given the same budget. We introduce a methodology that evaluates cost reduction using i) the eliminated shared cost in global plans and ii) the novel concept of cuts, that is, sets of materializations that exhibit synergy. We formulate the problem of choosing materializations for a target workload as a variant of the subexpression selection problem \cite{zhou2007efficient, jindal2018selecting}. We show that the materialization problem can be reduced, using cuts, into a Submodular Cover Submodular Knapsack (SCSCK) problem \cite{scskiyer}, for which there exists a family of approximation algorithms. Afterward, we address selecting which materializations to reuse and when in shared execution. We propose a reuse optimization pass that, at runtime, injects into a global plan the materialized subexpressions that maximize cost savings (i.e., eliminated computation minus filtering overhead) for the selected subexpressions.

\noindent \textbf{Access methods:} Filter amplification limits the applicability of reuse as it shrinks the net benefit and may even deteriorate performance. Efficient reuse requires that the processing time for shared filters over materializations is decreased. We reduce processing time for filters using suitable access methods for the workload at hand. By building and using access methods, \parcur\ enables shared execution to evaluate shared filters over one block of tuples at a time instead of processing them on a tuple-by-tuple basis, and thus to amortize the overhead. We build access methods for the target workload through partitioning and then use the created access methods to eliminate filters at runtime (Section \ref{parcur:sec:filter-skipping}).

\begin{table}[t]
\centering
\begin{tabular}{|c|c|c|c|c|}
\hline
                       & \rotatebox{90}{Shared cost\ } & \rotatebox{90}{Synergy\ } & \rotatebox{90}{Filter amplification\ } & \rotatebox{90}{Risk of miss\ } \\ \hline
\textcolor{brown}{Materialization policy}  &     \checkmark        &     \checkmark    &                      &              \\ \hline
\textcolor{brown}{Access methods}          &             &         &    \checkmark                  &              \\ \hline
\textcolor{brown}{Partitioning}            &             &         &                      &       \checkmark       \\ \hline
\textcolor{purple}{Reuse policy}            &      \checkmark       &      \checkmark   &                      &              \\ \hline
\textcolor{purple}{Data \& Filter skipping} &             &         &      \checkmark                &              \\ \hline
\textcolor{purple}{Partition-oriented execution}   &             &         &                      &        \checkmark      \\ \hline
\end{tabular}
\caption{Challenges (columns) and mechanisms (rows) that \parcur\ uses to harmonize reuse and work sharing. Brown rows: offline mechanisms, purple rows: online mechanisms.}
\label{parcur:table:components}
\vspace{-5mm}
\end{table}

\noindent \textbf{Partial reuse:} Strict subsumption limits the applicability of reuse. For this reason, \parcur\ opts for partial reuse: to exploit available materializations for the parts of the data that they cover. During execution, \parcur\ can answer each query by combining computations from parts of different materializations and even from parts of the base data. Computations on disjoint parts of the data that consist of filters, projections, join probes, and aggregations can be combined to produce the full result \cite{yang2021flexpushdown}. Our insight is that, to enable composable computations from different parts of data, planning and execution need to take place at partition-granularity. In addition, the reusability of materializations is maximum when they fully cover the data for a set of partitions. For those partitions, they always subsume the matching partition-local computations and can eliminate the corresponding processing. Hence, \parcur\ performs materialization and reuse at partition-granularity. The materialization policy selects for each materialization a set of partitions to fully cover and injects materializations into each partition's global plan at runtime. However, this scheme creates a dependency between partitioning and the storage overhead for covering the target workload; storage overhead is minimum when partition boundaries are aligned with the queries that the materializations subsume. Thus, due to this dependency, data needs to be partitioned such that each partition's tuples are required by the same computations, which, in turn, require the same materializations. We propose the metric of homogeneity to capture the similarity of computations across each partition's tuples. \parcur\ introduces a partitioning scheme that, by splitting data such that homogeneity is maximized, maps each computation to the data that it concerns and reduces wasteful materialization. \parcur\ uses the selected partitions at runtime in a partition-oriented execution model to enable partial reuse and achieves cost savings that are proportional to the overlap between the runtime and tuning workload.


\subsection{Putting It All Together}

We present \parcur{}, a framework that enables shared execution to effectively take advantage of materialized subexpressions by combining the proposed solutions. 

\parcur's architecture comprises two parts: the \emph{tuner}, and the \emph{executor}. The tuner operates offline. It analyzes a target workload made of historical query batches and adapts the framework's state by employing the \parcur's offline mechanisms: i) it partitions the data based on the access patterns of the target workload, ii) it materializes a set of subexpressions for the given partitions, and iii) it builds new access methods for the materialized subexpressions using finer-grained partitioning. Then, given the available partitioning, materialized subexpressions, and access methods, the executor processes each query batch arriving at runtime: i) it performs shared execution at the level of the available partitions, ii) it decides when and where to reuse materialized subexpressions for each partition, and iii) it uses the available access methods to reduce filter costs using data and filter-skipping. Figure \ref{parcur:fig:workflow} illustrates the end-to-end workflow for both the offline tuner and the online executor. We elaborate on each of these mechanisms in Sections \ref{parcur:sec:access-tuning} and \ref{parcur:sec:learning}. Note that the query batches processed at runtime can be arbitrarily different from historical batches both in terms of access patterns and global plans. In all cases, \parcur\ opportunistically uses the existing state to reduce the response time of runtime batches.

\begin{figure}[t]
\begin{tabular}{@{}l@{}l}
\begin{minipage}{0.5\hsize}
    \centering\includegraphics[width=\hsize, page=1, trim=0 0 0 0, clip]{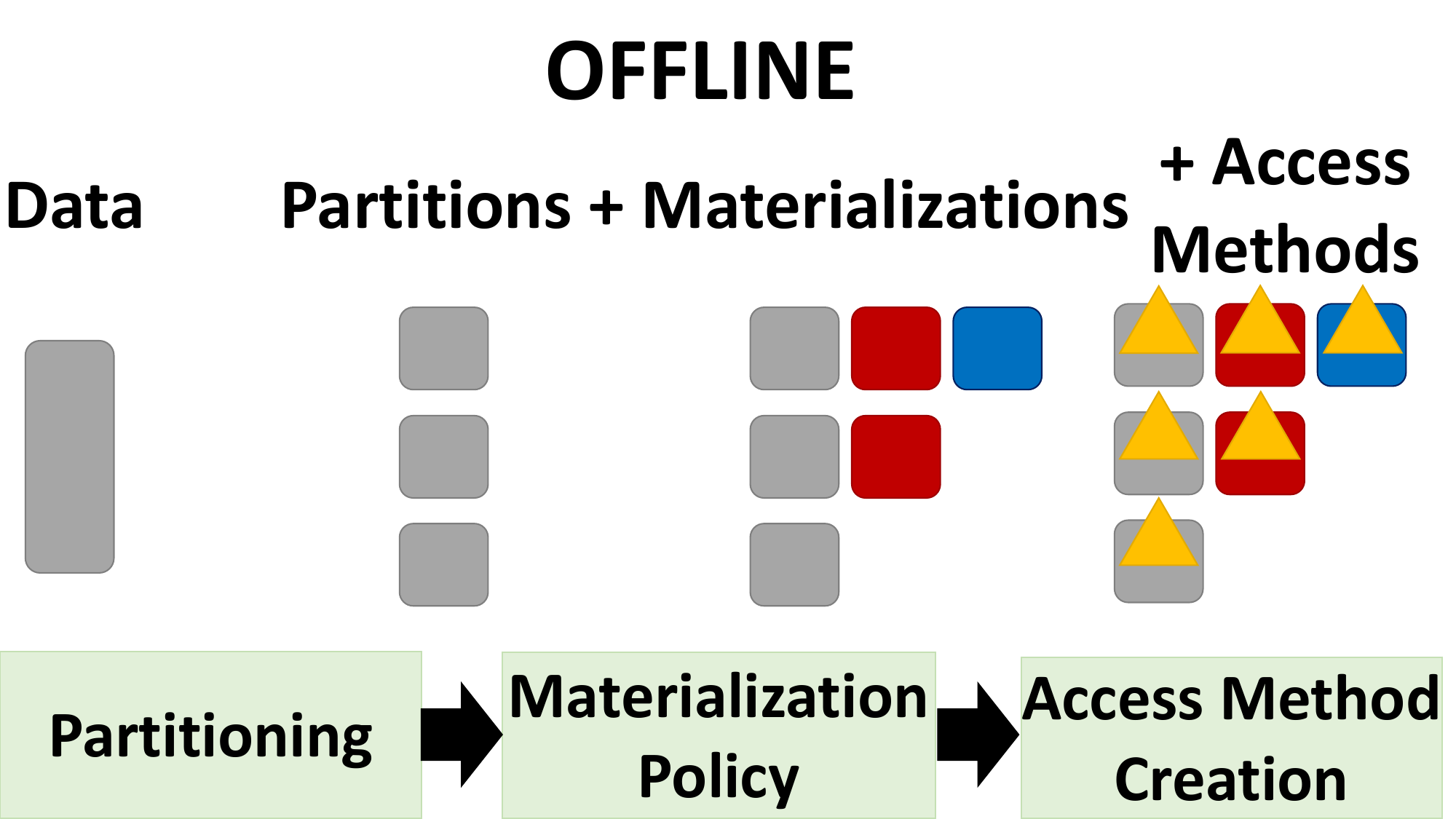}	\\
    \vspace{-2mm}
        \centering {\footnotesize(a)}
        \end{minipage}
        &
    \begin{minipage}{0.5\hsize}
\centering\includegraphics[width=\hsize, page=2, trim=0 0 0 0]{parcur/figures/workflow.pdf}	\\
\vspace{-2mm}
    \centering {\footnotesize(b)}
    \end{minipage}
\end{tabular}
\vspace{-4mm}
\caption{\parcur's workflow in a) the offline tuner and b) the online executor}
\vspace{-5mm}
\label{parcur:fig:workflow}
\end{figure}

\vspace{-2mm}

\section{Tuning \parcur's state}
\label{parcur:sec:access-tuning}

By analyzing a target workload that consists of a sequence of query batches, the tuner repartitions the data, materializes a set of selected subexpressions, and builds access methods on the materialized subexpressions. Tuning takes place offline. After tuning is done, the partitions, the materialized subexpressions and the access methods are exposed to the executor at runtime, which uses them to eliminate recurring computation in subsequent query batches.


In this section, we present each of the steps in the tuner's workflow. Each step's output is the input for the next step in line: partitioning chooses the boundaries for materializing subexpressions and the materialization policy selects the subexpressions on which to build access methods. We first present the partitioning algorithm (Section \ref{parcur:data-layout}), then introduce the materialization policy (Section \ref{parcur:view-selection}) and finally discuss access method construction (Section \ref{parcur:sec:building-access}).

\vspace{-2mm}

\subsection{Workload-driven Partitioning}
\label{parcur:data-layout}

The first step of \parcur's tuner is to partition the data in a way that maximizes the utility of subsequent materializations. To differentiate between this partitioning and any additional data reorganization for building access methods, we name the first step's partitioning as \emph{1st-level partitioning} and any further partitioning as \emph{2nd-level}.

\begin{figure}[t]
 \centering
	\centering\includegraphics[width=0.7\hsize, page=1, trim=0 250 0 0, clip]{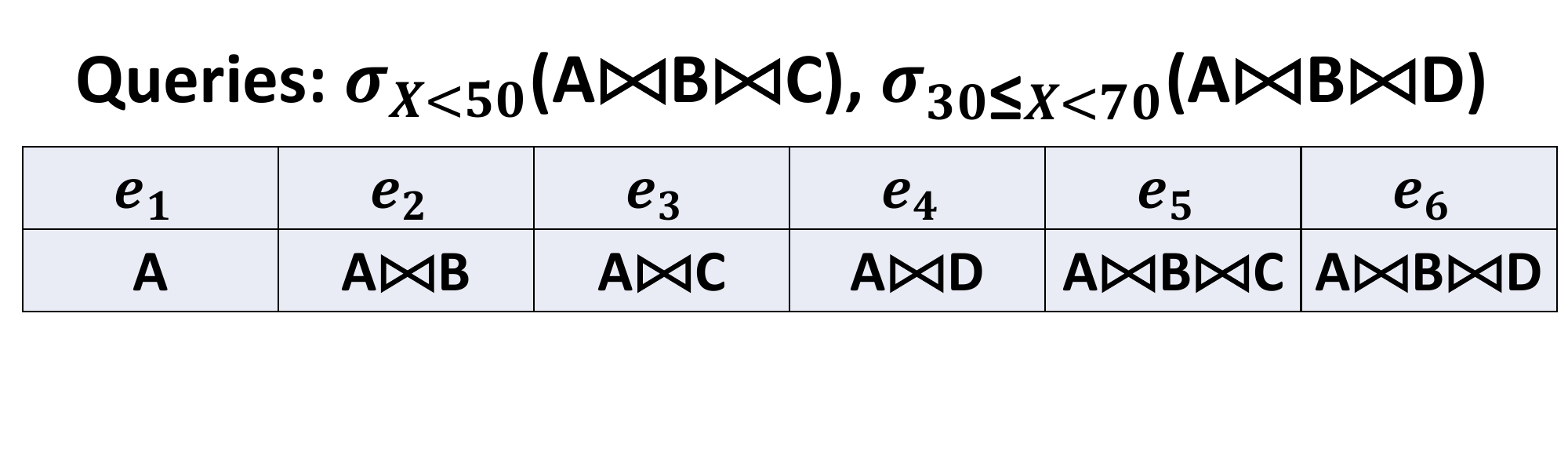} \\
 \vspace{-6mm}
    \centering{\footnotesize(a) Subqueries in a two-query batch} \\
    \includegraphics[width=0.8\hsize, page=3, trim=0 270 0 0]{parcur/figures/clustering-example.pdf}	
     \\
     \vspace{-3mm}
    \centering{\footnotesize(b) Subquery vector for tuples in $[0, 30), [30, 50)$}
    \vspace{-3mm}
	\caption{Two-query example for subquery vectors}
 \vspace{-6mm}
	\label{parcur:fig:clustering-example}
\end{figure}

For partition-granularity materialization to be budget-efficient, all tuples should be processed by similar query patterns, i.e., most of their downstream computation should be the same. Therefore, \parcur\ employs a partitioning scheme that clusters tuples according to query patterns and materializes subexpressions for each partition independently.


Such a partitioning scheme offers three benefits: i) if the query patterns remain the same, materialized subexpressions are almost fully reused, and space is not wasted, ii) materialization is specialized for the sharing decisions of each partition's query pattern, and iii) for the case of \emph{partial reuse} during a workload shift, performance degradation becomes proportional to the magnitude of the shift.


To cluster together data that is processed by similar query patterns, we keep track of processing history for a sample of tuples by maintaining a \emph{subquery-vector} for each tuple. We consider all possible subqueries $e_1, e_2, \dots, e_m$, that appear in a set of historical batches, and mark to which of them each tuple belongs. By \emph{subqueries}, we mean all the join subexpressions (and their reorderings), that exist in each batch, and involve the fact table. For example, considering a batch with two queries, $A \bowtie B \bowtie C,\ A \bowtie B \bowtie D$ and $A$ as the fact table, leads to the subqueries depicted in Figure \ref{parcur:fig:clustering-example}a. We represent subexpressions in different batches as separate subqueries because they do not actually co-occur. Using subqueries is advantageous as it exposes similarities that do not depend on a specific execution plan and naturally represents co-occurrence in the same batch.

We then use the subquery-vectors in order to formulate a tuple-clustering problem based on homogeneity. We assume a matrix $W$, where the $i_{th}$ row of it corresponds to the subquery-vector of the $i_{th}$ tuple: If at least one query with subquery $e_j$ accesses the $i_{th}$ tuple, $W_{i,j}=w(e_j)$, where $w(e_j)$ is a weight assigned to $e_j$. Otherwise, $W_{i,j}=0$. In our implementation, to increase the relative importance of larger subqueries to homogeneity, we set $w(e_j)=|e_j|$, where $|e_j|$ is the number of tables participating in $e_j$. Alternatives assignments can also achieve a similar result. 

Given a set of tuples $T$, we formally define homogeneity as:

$$
H(T, W) = \displaystyle\sum_{t\in T} \frac{\sum_{j=1}^m W_{t,j}}
{max(\sum_{j=1}^m w(e_j) \times u(\sum_{t\in T} W_{t,j}), 1)}
$$

where $u(x)$ is the step function with $u(x)=1$ when $x>0$ and $0$ otherwise. Homogeneity assigns a score to each tuple in $T$ based on the subqueries that access the tuple and is defined as the sum of scores. Each tuple's score is the sum of weights for the subqueries that access tuple $t$ over the sum of weights for the subqueries that access at least one tuple in $T$. Hence, the complexity for computing $H(T,W)$ is $O(m \times |T|)$. The score is maximum (i.e., equals $1$) if all the subqueries that access at least one tuple in $T$ also access $t$. The intuition is that homogeneity is maximum when all tuples in $T$ are accessed by the exact same subqueries. In that case, the utilization of materializations is also maximum; assuming that a subquery's results are materialized and that the historical batches recur as is, reuse exploits all the tuples in the materialization, and no tuple is redundant.

Homogeneity-based partitioning is defined as finding the partitions $\{ p^*_1, p^*_2, \dots, p^*_n \}$ that maximize the aggregate homogeneity:
\vspace{-0.2in}
$$
\{ p^*_1, p^*_2, \dots, p^*_n \} = \argmax_{ \{p_1, \dots, p_n\} } 
 \sum_{i=1}^n 
H(p_i, W) \ s.t.\ \forall p_i\ |p_i| \geq PS_{min}
$$
where $PS_{min}$ is the minimum allowed partition size. Homogeneity-based partitioning finds partitions such that, in each partition, the tuples are accessed by almost the same set of subqueries, and thus, barring a workload shift, the utilization of materializations is high. The partition size constraint ensures that the solution avoids the trivial optimal solution where each tuple forms its own partition.



%




To efficiently compute a solution to homogeneity-based partitioning, we use a space-cutting approach that, similar to \cite{yang2020qd}, forms a tree of cuts in the space of table attributes. Each internal node corresponds to a logical subspace of the table and contains a predicate based on which this subspace is further split: the left child corresponds to the data that satisfies the predicate, whereas the right child to the data that does not. Finally, the leaves of the tree correspond to data partitions, which are the quanta for materialization. The advantage of the space-cutting approach is that it enables routing queries to required partitions based on the predicates of the splits and the queries.


\begin{algorithm}[t]
\SetAlgoLined
\caption{Homogeneity-based Partitioning}
\label{parcur:alg:homogeneity-based clustering}
\SetKwProg{Fn}{Function}{ :}{}
\Fn{PARTITION($partition, W, cuts, PS_{min}$)}{
$output = null$ ; $best = null$ ; $bestScore = null$ \;
$score = H(partition.sample, W)$ \; 
\For {$cut \in cuts$} {
    \If {$intersects(partition, cut)$} {
        $tp, fp = getPartitions(partition, cut)$ \;
        \If {$tp.size < PS_{min}\ or\ fp.size < PS_{min}$} {
            $continue$ \;
        }
        $curr = H(tp.sample,W) + H(fp.sample, W)$ \;
        \If {$best == null\ or\ curr > bestScore$} {
            $best = cut$ ; $bestScore = curr$ \;
        }
    }
}
\If {$best == null\ and\ bestScore > 1.01\times score$} {
    $tp, fp = getPartitions(partition, best)$ \;
    $tres = PARTITION(tp, W, cuts, PS_{min})$ \;
    $fres = PARTITION(fp, W, cuts, PS_{min})$ \;
    $output = Node(best, tres. fres)$ \;
}
\lElse {
    $output = Leaf()$ 
}
}
\Return $output$ \;

\end{algorithm}
\setlength{\textfloatsep}{0ex}

To solve the partitioning problem, we use the greedy Algorithm \ref{parcur:alg:homogeneity-based clustering}. The algorithm runs on a uniform sample of the tuples to keep runtime monitoring overhead low. \parcur\ computes the sample's query pattern matrix by monitoring data accesses across batches and by recording the vector of subqueries for the sample's tuples. When triggered, the greedy algorithm computes the change in the objective function for each candidate cut, that is, a predicate that intersects with the partition at hand (lines 5-9), and finds the locally optimal cut that maximizes the aggregate homogeneity (lines 10-11). Then, the space is partitioned based on the locally optimal cut, and the greedy algorithm is recursively invoked for the two children subspaces and the respective sample tuples (lines 13-16). The greedy algorithm stops when either the relative improvement from the locally optimal cut drops below a threshold, which we set at 1\% (line 7), or all candidate cuts violate the minimum partition size for resulting partitions (lines 7-8).

Let $S_p$ the sample per partition and $S$ the entire sample. The algorithm's complexity  depends on i) the number of recursive invocations, ii) the complexity of $H(S_p,W)$, which is $O(m \times |S_p|)$, and iii) the number $|F|$ of distinct filters in the tuning workload. As the minimum partition size is $PS_{min}$, we can have at most $\frac{|S|}{PS_{min}}$ leaf-partitions, and hence $\frac{2 \times |S|}{PS_{min}}-1=O(\frac{|S|}{PS_{min}})$ invocations. Thus, the complexity of Algorithm \ref{parcur:alg:homogeneity-based clustering} is: $O(\frac{|S|}{PS_{min}} \times m \times |S| \times|F|)$.

Homogeneity-based partitioning results in more efficient use of the storage budget compared to data access-based partitioning schemes, such as Qd-tree \cite{yang2020qd}.




\vspace{-3mm}

\subsection{Materialization Policy}
\label{parcur:view-selection}

1st-level partitioning assumes that query patterns represent the overall workload and thus recur in future batches. To eliminate recomputation in such cases, \parcur\ materializes subexpressions on a per 1st-level partition basis.
 
Due to the interference between reuse and work sharing, a global-plan-aware materialization policy is required. Also, in \parcur\, the policy should consider that partitions process different query patterns. Hence, \parcur\ relies on a new formulation of the subexpression selection problem, which is: i) sharing-aware, and ii) works on \emph{partition-wise global plans}. The optimal solution differs from the one of the classical problem. We call this new problem \emph{Multi-Partition Subexpression Selection for Sharing (MS3)}. We define the \emph{Historical Workload Graph}, which is the input of MS3, and then MS3 itself.


\begin{defn}[Historical Workload Graph]
\emph{Given a fact table $T$, a partitioning $\{p_1, p_2,..,p_n\}$ of $T$, and batches $\{Q_1, Q_2, \dots, Q_m \}$, the historical workload graph $G$ is a graph composed of connected components $G_{i,j}$, $i \in \{1,\dots, n\}$, $j \in \{1,\dots, m\}$, where $G_{i,j}$ is the global plan for $Q_j$ over $p_i$. In the global plan, nodes represent operators (including a pseudo-operator for $T$) and edges represent producer-consumer relationships.}
\end{defn}

\begin{defn}[MS3]
\emph{Let $R(c)$ be the maximum cost reduction that reuse can incur when executing the global plans of the historical workload graph $G$ with an available set of materialized subexpressions $c$, and $B(c)$ the budget required for materializing $c$. If $\mathcal{B}$ is the total memory budget, MS3 is defined as:}

$$
\max_{ c } R(c), \text{s.t.: } B(c) \leq \mathcal{B}
$$

\end{defn}


MS3 is a hard problem and hence it is time-consuming to compute a tractable exact solution. To solve it, we first prove a reduction to \emph{Submodular Cover Submodular Knapsack (SCSK)} problem \cite{scskiyer} and then show how we can use approximate algorithms for SCSK to choose to materialize a set of expressions that achieve a high cost reduction with approximation guarantees.

\subsubsection{Reduction to SCSK}

Let $U$ be a set and $f, g:2^U \rightarrow \mathbb{R}$ be two submodular functions\footnote{Submodularity formalizes diminishing returns. Specifically, a function $h$ is defined as submodular if 
$S \subset S' \Rightarrow h(S \cup \{ s \}) - h(S) \geq h(S' \cup \{ s \}) - h(S')$.},
then SCSK is the optimization problem
$$
\max_{S\subset U}\ g(S),\ s.t.\ f(S) \leq B
$$

To reduce MS3 to SCSK, cost savings in MS3 should be submodular, i.e., adding more materialized subexpressions should result in diminishing returns. While this holds in QaT execution, where each materialization reduces the marginal benefit of other conflicting materializations, it does not hold in shared execution. We observe that shared execution benefits more from materializations in the same path of the global plan where synergy increases cost savings.

The key idea for reducing MS3 to a submodular optimization problem is to materialize subexpressions in groups. We notice that computing cost savings for groups gives us more accurate estimates for the eliminated upstream computations. In addition, synergy between groups is always captured by their \emph{super-group}, i.e., a group that contains their union.

We formulate useful groups of materializations by introducing the concept of \textit{cuts}. Intuitively, in a given global plan, a cut is a set of subexpressions that, if materialized, eliminate all upstream operators between (inclusive) the operators that produce them and a common ancestor, the \textit{anchor}. For example, in Figure \ref{parcur:fig:motivation-new}, the cut composed of operators 3 and 8 also eliminates the upstream operators 1 and 2, which are also anchors. Formally, we define cuts and anchors as follows:

\begin{defn}[Cuts and anchors]
Let $G$ be the historical workload graph. A set of nodes $c \subset V$ is defined as a cut with respect to anchor $a\in V$ if:
\begin{itemize}
    \item $a$ is an ancestor of every $v\in c$.
    \item Every descendant of $a$ is either i) an ancestor of at least one node $v \in c$, or ii) a descendant of exactly one node $v \in c$.
\end{itemize}
We represent the set of all cuts in $G$ as $CUTS(G)$, and for all $c\in CUTS(G)$ we define $BC(c, a)$ as the nodes between (inclusive) the cut's nodes and anchor $a$. The shorthand $BC(c)$ implies using the minimal anchor (i.e., an anchor whose predecessor is not an anchor for $c$).
\end{defn}


Choosing cuts to materialize so as to maximize the eliminated cost in their BC sets is related but not identical to $MS3$. The cost in $BC$ sets is not always equal to the cost reduction from the same materializations in $MS3$, because $MS3$ implicitly includes the savings from super-cuts, that is the union of smaller materialized cuts. However, we prove that solutions in the \textit{cut selection} problem can be enriched such that they are both solutions to \textit{cut selection} and $MS3$ with equal savings. Furthermore, we prove that \textit{cut selection} is an SCSK problem, and thus we can solve it using approximate algorithms. Based on these two properties, \textit{cut selection} gives a solution to $MS3$ with better or equal approximation factor than the one given for SCSK. In the following paragraphs, we formally define \textit{cut selection} and prove the mentioned properties.

First, we introduce some required notation:

\begin{defn}[Domain and Enrichment]
Let $S$ be a set of cuts to materialize. We define the \textit{domain} of $S$ as the set 
\vspace{-1mm}
$$d(S)=\{v.subquery | v \in (\bigcup_{c \in S}\ c) \}$$
\vspace{-1mm}
and the \textit{enrichment} of $S$ as the set 
\vspace{-1mm}
$$ e(S) = \{ c | c \in CUTS(G)\ and\ \forall v\in c  (v.subquery\in d(S)) \} $$
\vspace{-1mm}
The domain represents which results $S$ materializes, and the enrichment represents all cuts that are materialized by materializing $S$.
\end{defn}

\begin{defn}[Cost Reduction and Budget]
Let $S$ be a set of cuts. Also, let $cost(op)$ be the processing cost for an operator $op$ in the global plan. We model the cost reduction due to materializing $S$ as:

$$
\bar{R}(S) = \sum_{op\in O} cost(op),\ where\ O=\bigcup_{c\in S} BC(c)
$$

and the required materialization budget as $\bar{B}(S) = \sum_{v\in d(S)} B(\{v\})$.
$\bar{R}(S)$ is equal to the cost of operators $O$ that are eliminated by materializing the cuts in $S$. Each cut eliminates the shared operators between the cut and the minimal anchor and, by definition, computing $O$ as the union of operators accounts for overlaps between the operators that are eliminated by different cuts.

$\bar{B}(S)$ equals the total budget required for materializing the results of the cuts in $S$. $d(S)$ is by definition the results that $S$ materializes.
\end{defn}

\begin{defn}[Reduced Workload Graph]
Let $S$ be a set of cuts. We define the \textit{reduced workload graph} of $S$ as
$$
G(\emptyset)=<V(\emptyset), E(\emptyset)>=G
$$
$$
G(S)=<V(S), E(S)> = G[V - \bigcup_{c \in S} BC(c)]
$$
where $G[V']$ is the induced subgraph of $G$ for vertices $V'$.

The reduced workload graph represents the global plans for the historical query batches after materializing and reusing the cuts in $S$.

\end{defn}

We define \textit{cut selection} problem as follows:

\begin{defn}[Cut Selection]
Cut selection is defined as the optimization problem of finding a set of cuts $S$ such that:
$$
max \bar{R}(S), \text{s.t.: } \bar{B}(S) \leq B
$$
\end{defn}


Using the above definitions, we prove the following theorems:

\begin{theorem}
Cut selection is a SCSK problem.
\end{theorem}

\begin{proof}
 We prove that $\bar{R}$ and $\bar{B}$ are submodular. For a set of cuts $S$ and a cut $c$, it holds that:
$$
\bar{R}(S \cup \{c\}) - \bar{R}(S) = \sum_{op\in O(S)} cost(op)\ s.t.\ O(S)=BC(c) \cap V(S)
$$
and
$$
\bar{B}(S \cup \{c\}) - \bar{B}(S) = \sum_{m\in M(S)} B(\{m\})\ s.t.\ M(S)=d(\{c\}) - d(S)
$$

Let S, S' be two sets of cuts such that $S \subset S'$. Then,

$$
\bar{R}(S \cup \{c\}) - \bar{R}(S) = \sum_{op\in O(S)} cost(op)
$$
and
$$
\bar{R}(S' \cup \{c\}) - \bar{R}(S') = \sum_{op\in O(S')} cost(op)
$$
However, $V(S') \subset V(S)$ \textit{and thus} $O(S') \subset O(S)$. Therefore,
$$
\bar{R}(S \cup \{c\}) - \bar{R}(S) \geq \bar{R}(S' \cup \{c\}) - \bar{R}(S')
$$


Similarly,

$$
\bar{B}(S \cup \{c\}) - \bar{B}(S) = \sum_{m\in M(S)}
$$
and
$$
\bar{B}(S' \cup \{c\}) - \bar{B}(S') = \sum_{m\in M(S')}
$$
Then, $d(S) \subset d(S')$ \textit{and thus} $M(S') \subset M(S)$. Therefore,
$$
\bar{B}(S \cup \{c\}) - \bar{B}(S) \geq \bar{B}(S' \cup \{c\}) - \bar{B}(S')
$$


Therefore, both f and g are submodular.
\end{proof}

\begin{theorem}
If $S$ is a solution to cut selection, then $e(S)$ is also a solution to cut selection with $\bar{R}(e(S)) \geq \bar{R}(S)$.
\end{theorem}

\begin{proof}

By definition, $S \subset e(S)$. It also holds: \\ $\bar{B}(e(S)) = \sum_{m \in d(S)} B(\{m\}) = \bar{B}(S)$.
Therefore $\bar{B}(e(S)) \leq B$ and $e(S)$ is a solution to cut selection.

Furthermore, $\bar{R}(e(S)) = \sum_{op\in (\bigcup_{c \in e(S)} BC(c))}$
However, $S \subset e(S) \Rightarrow \displaystyle\bigcup_{c \in S} BC(c) \subset \bigcup_{c \in e(S)} BC(c)$. So, $\bar{R}(e(S)) \geq \bar{R}(S)$.
\end{proof}

\begin{theorem}
For every $e(S)$, it holds that $R(d(S)) = \bar{R}(e(S))$
\end{theorem}

\begin{proof}
Let $S$ be a set of cuts. We represent the eliminated operators in the original MS3 problem when $S$ is materialized as $t(S)$. Formally, $t(S)$ is the set of all nodes whose operators produce $d(S)$ or all their successors belong to $t(S)$. Then:

$$
R(d(S)) = \sum_{op \in t(S)} cost(op)
$$

We now prove that $t(S) = \bigcup_{c \in e(S)} BC(c)$.

Let $c'\in e(S)$. Then, $c' \subset d(S)$, and $\forall v\in BC(c')$ it holds $v \in t(S)$ and thus $BC(c') \subset t(S)$. It follows that $\displaystyle\bigcup_{c' \in e(S)} BC(c') \subset t(S)$.

Also, let $a\in t(S)$ and $c_a$ all the descendants of $a$ that belong to $d(S)$. Then, $c_a$ is a cut with anchor $a$, as the two conditions in the definition of cuts are true: i) $a$ is an ancestor for all nodes in $c_a$, and ii) assume there is a descendant of $a$, $a'$, that is not a descendant of any node in $c_a$. Then, $a'$ is an ancestor of at least one node in $c_a$ because $a\in t(S)$ (otherwise, the nodes in the path from $a$ to $a'$ should not be in $t(S)$). Therefore, $c_a$ is a cut, $a \in BC(c_a)$ and $t(S) \subset \displaystyle\bigcup_{c' \in e(S)} BC(c')$.

Thus $t(S) = \displaystyle\bigcup_{c \in e(S)} BC(c)$ and:

$$
R(d(S)) = \displaystyle\sum_{op \in t(S)} cost(op) = \displaystyle\sum_{op \in \bigcup_{c \in e(S)} BC(c)} cost(op) = \bar{R}(e(S))
$$
\end{proof}

\subsubsection{Approximating MS3}

\parcur's tuner chooses subexpressions to materialize by solving cut selection for historical batches. The selection process has two steps: i) the tuner constructs the workload graph and computes the cuts and their corresponding $BC$ sets, ii) the tuner runs an algorithm for solving the cut selection instance for the computed cuts. The subexpressions in the selected cuts are then materialized and used in subsequent batches.

The tuner currently implements two approximate algorithms for solving SCSK, greedy (Gr) and iterative submodular knapsack (ISK) \cite{scskiyer}. We briefly present the properties of the two algorithms as presented in the work of Iyer et al. \cite{scskiyer}.

\noindent \textbf{Gr}:  Gr is a greedy algorithm. At each step, it chooses the cut with the highest marginal benefit that can fit in the budget and adds it to the solution. Gr's complexity is $O(|CUTS(G)|^2)$, and in practice it takes few msecs. Gr provides an approximation factor:
$
1-(\frac{K_f - 1}{K_f})^{k_f}
$,
where $ K_f=\{max_{S\subset U}\{|S| | f(S) \leq B\}$ and $k_f=\{min_{S\subset U}\{|S| | f(S) \leq B \wedge f(S\cup\{j\}) > B \}$. Indeed, Gr is inefficient when few cuts can saturate the budget.

\noindent \textbf{ISK}: ISK is a fixed point algorithm. In each iteration, it combines partial enumeration with greedy expansion; it chooses between $\binom{|CUTS(G)|}{3}$ candidate solutions, where each candidate fixes the first three cuts and chooses the rest using a greedy algorithm. At each step, the greedy algorithm chooses the cut with the highest ratio of marginal benefit to required budget. The solution in each iteration affects the budget calculation for the next iteration. ISK's complexity is $O(|CUTS(G)|^5)$ and can run for hundreds of seconds for a few hundreds of cuts. ISK provides a constant factor $1-e^{-1}$ for the solution of $\{max_{S\subset U} g(S) | f(S) \leq \frac{b}{K_f}\}$ and a bicriterion guarantee if we run it with a larger budget constraint \cite{scskiyer}.





\vspace{-0.1in}
\subsection{Building Access Methods}
\label{parcur:sec:building-access}

At runtime, materialized subexpressions are accessed at a per-partition level. Nevertheless, they still need to be scanned and filtered based on the predicates of the running queries. The processing time for shared access and filtering of base and cached data can dominate the total processing time. \parcur\ further reduces both data access- and filtering costs by reorganizing data within each partition using multidimensional range partitioning. We refer to this finer-grained partitioning as \emph{2nd-level partitioning}.


Multidimensional range partitioning can enable efficient data access that reduces accesses during scans, as it enables data skipping. Furthermore, by cutting data across values that are frequently used in predicates, it can be used to statically evaluate frequent filters for a whole partition. To build the partitions, we iteratively subpartition data across the predicates values of one attribute at a time. The resulting subpartitions inherit query homogeneity from the \emph{1st-level} partitioning and also reduce data-access costs. From this point on, we differentiate the partitions derived from the \emph{2nd-level} partitioning by calling them \emph{blocks}.
\vspace{-0.1in}
\section{Reuse-aware Shared Execution}
\label{parcur:sec:learning}

At execution time, \parcur\ takes advantage of the constructed partitions and materialized subexpressions and optimizes query processing in three levels: First, it uses data and filter skipping to identify the queries that access each partition and reduces filtering costs. Second, it adopts a partition-oriented execution paradigm that plans and optimizes each partition independently; thus, exposing different opportunities per partition. Third, \parcur\ introduces a cost-based optimization framework that chooses which materializations to inject into each partition's plan.


\subsection{Data and Filter Skipping}
\label{parcur:sec:filter-skipping}

\parcur\ uses 2nd-level partitioning to reduce data access and filtering costs. To do so, for each block, it identifies i) which queries process the block, and ii) which predicates have the same value for all tuples in the block. Then, during execution, it skips 2nd-level partitions that are not processed by any query, and eliminates filters whose predicates are invariant across the block. Both optimizations occur on both the fact table and the materializations, and can drastically reduce batch response time.

As the data is organized by cutting the data space, each block's boundaries are defined by a range along each attribute. Then, if the range is known, the above analysis can be done statically. Concretely, a query's predicate is invariant when its value range either subsumes (always true) or does not overlap (always false) with the block's range. Moreover, a query \emph{skips a block} if at least one of its predicates always evaluates to false (no overlap).  For example, the query \texttt{SELECT COUNT(*) FROM T WHERE x > 8} skips block $5 \leq x < 7$, as the two ranges do not overlap. Similarly, for the same block, the predicate of query \texttt{SELECT COUNT(*) FROM T WHERE x > 4} is true across the whole block and, thus, it is redundant to evaluate it for every tuple.

The above logic is implemented by maintaining zone-maps \cite{graefe2009fast}: a lightweight index that stores min-max statistics for each attribute. During the table scan, for each block, \parcur\ compares the corresponding ranges against the shared filter predicates to identify which queries do not overlap with this block (data skipping), and which are satisfied by the entire block (filter skipping).  The remaining ambivalent filters are processed using the global plan.

\vspace{-0.1in}
\subsection{Partitioned Execution}
\label{parcur:sec:partitioned-execution}

\parcur\ optimizes each 1-st level partition independently to i) exploit partition-specific materializations, and ii) enable partial reuse by decoupling planning between partitions. To do so, it introduces a two-phase partition-oriented execution model. First, it computes the shared state between partitions such as hash tables on dimensions and data structures for aggregation. Next, it executes each partition independently. For each partition, \parcur\ identifies which queries process the partition using the same data-skipping mechanism as above. Then, it chooses a global plan that is specialized for the queries and materializations of the partition at hand. Finally, partial results from each partition are merged together in the output operators such as projections, aggregations, and GROUP-BYs. To reduce aggregation overheads, our implementation preaggregates partial results at the thread-level. Since shared execution processes subqueries that comprise selection, projection, join probe, and potentially aggregation operators, combining partial results produces the final output \cite{yang2021flexpushdown}.

Partial reuse is feasible because the output operators are oblivious to each partition's planning decisions. When query patterns recur with minor shifts, they mostly process their designated 1-st level partitions and spill over only to few neighboring partitions. Then, \parcur\ processes the bulk of the processing using materializations and addresses spill-overs with selective computations. Hence, in case of a workload shift, performance degradation becomes proportional to the magnitude of the shift, and thus \parcur\ avoids suffering a performance cliff.

\subsection{Injecting Materializations in Global Plans}
\label{parcur:view-aware-sharing}

For each partition, \parcur\ optimizes and processes a global plan that exploits the available materializations and access methods as well as sharing opportunities. However, making all planning decisions in a unified optimization framework scales poorly. To this end, \parcur\ adopts a two-phase optimizer: it first finds the best possible global plan that only uses work-sharing, and then, it improves it by optimally substituting shared operators with materialized views.


\subsubsection{Two-phase optimizer}
\label{parcur:sec:two-phase-opt}

Benefits from reuse and work sharing are interdependent: the marginal benefit from reuse, if any, depends on available sharing opportunities and, also, the opportunities from downstream work sharing between queries are contingent on answering them using the same materialization. Thus, it is tempting to formulate a unified optimization problem in order to find a globally optimal plan. However, sharing-aware optimization already has a very large search space, and thus enriching it with reuse planning decisions is prohibitive.

To incorporate work sharing and reuse in a scalable and practical manner, the optimizer needs to restrict the search space. \parcur's optimizer focuses on ensuring better performance than pure work sharing and on avoiding performance regression. Thus, the optimizer uses two phases. In the first phase, the optimizer chooses a \textit{baseline} global plan that uses work sharing. Then, in the second phase, the optimizer improves on the baseline plan by rewriting it to reuse materializations. Finally, \parcur\ processes the resulting plan, which combines reuse and work sharing.

\subsubsection{Reuse phase}
\label{parcur:sec:injecting-reuse}

The reuse phase is based on the observation that reuse replaces operators from the baseline plan with filters on materializations. Hence, the goal is to find which subexpressions, if reused, can maximize the difference between eliminated computations and filtering costs. For each cut $c$, we can estimate this difference, which we call \textit{benefit}, as:
$$
benefit(c, a) = \sum_{op\in BC(c, a)} cost(op) - \sum_{v\in c} (c_f \times |RF(v)| \times v.size)
$$
where $cost(op)$ of operator $op$ in the baseline plan, $RF(v)$ are the runtime filters on subexpression $v$ after filter-skipping in the current partition, $v.size$ is the number of tuples for the subexpression in the current partition and $c_f$ is a constant for estimating filtering costs per tuple as a linear function of the number of runtime filters $|RF(v)|$. $benefit(c, a)$ represents the net benefit of reusing $c$ with respect to anchor $a$ as the different between the cost of eliminated operators between $c$ and $a$ and the overhead for accessing and filtering $c$'s materializations. The optimizer has all this information at the time of running the reuse phase.

In order to choose which subexpressions to reuse, the reuse phase, which we show in Algorithm \ref{parcur:alg:reuse-phase}, performs a post-order traversal of the baseline plan and transforms the plan. When visiting a node, the traversal first processes the node's successors and merges their rewrite decisions (lines 4-6). Then, the algorithm finds the best cut (i.e., the cut with the highest benefit) that can eliminate the current node. If all of the node's successors are eliminated or are anchors for cuts, then the algorithm computes the best cut of downstream subexpressions by merging the cuts of the remaining successors (lines 9-12). If the node corresponds to a materialized subexpressions, the algorithm also considers the cut that consists of the node's results (lines 15-16). Finally, if the best cut provides net gain, the rewrite is applied immediately (lines 17-19), and otherwise the best cut is propagated to upstream nodes.

\begin{theorem}
Given a plan, Algorithm \ref{parcur:alg:reuse-phase} makes optimal view injection.
\end{theorem}

\begin{proof}
By induction on the plan size. \textbf{Base step:} Single-node plan. If reuse is beneficial, the plan is rewritten. Otherwise, it is optimal and stays as is. \textbf{Induction step:} If it holds for plan size $\leq k$, it also holds for $k+1$. We assume a single root in the plan. If the plan consists of multiple connected components, then separately solving for each component is trivially optimal. The first visited node is the root. Each downstream subplan $\leq k$ nodes, so the algorithm minimizes the cost. Let each node have an attribute $optPlan$ that represents the optimal downstream plan and $DC(plan)$ be the cost for an optimized downstream plan. Before line \ref{line:alproof}, we have:

\setlength{\textfloatsep}{0ex}
$$
\Delta = \sum_{s \in succ} (DC(s.optPlan) - DC(s.bestPlan)) + (x-y) cost(v)
$$

where $x,y\in\{0,1\}$ represent if $v$ is part of the plan. Then, there are the following two cases:
   
Case1: if $x=1$ or $y=0$, then $\Delta \geq 0$. Thus, $bestPlan$ is optimal.

Case2: if $x=0$ and $y=1$, we prove that the algorithm eliminates $v$ in line \ref{line:alproof} and $\Delta \geq 0$ for the new $bestPlan$. Since $x=0$, there exists a cut $c$ with anchor $v$.

- If $benefit(\{v\},v) > 0 \Rightarrow \Delta \geq 0$ for the new $bestPlan$.

- Otherwise, REWRITE needs to happen in the downstream cut. Let $s_1, s_2, \dots, s_p$ be $v$'s successors and $c_1, c_2, \dots, c_p$ the corresponding sub-cuts. Since $optPlan$ is optimal: $\sum  benefit(c_i,s_i) + cost(v) \geq 0$, where $i \in \{i|benefit(c_i,s_i) < 0 \}$.
Thus, the merged cuts from the successors can eliminate $v$ and the new $bestPlan$ is optimal.

\end{proof}

\begin{algorithm}[t]
\SetAlgoLined
\caption{Reuse Optimization Phase}
\label{parcur:alg:reuse-phase}
\SetKwProg{Fn}{Function}{ :}{}

\Fn{REUSE\_OPT\_REC($v$)}{
    \label{line:reuse-phase-rec}
    $v.bestPlan=\emptyset$ ; $v.bestCut= (!v.succ.empty()) ?\ \emptyset : null$ \;
    $atLeastOne = (v.succ.empty())$ \;
    \For {$s \in v.succ$} {
        $REUSE\_OPT\_REC(s)$ \;
        $v.bestPlan =  v.bestPlan \cup s.bestPlan$ \;
        \If {$s.bestPlan.contains(s)$} {
            $atLeastOne = true$ \;
            \If {$v.bestCut != null$} {
                \If {$s.bestCut == null$} {
                    $v.bestCut = null$ \;
                }\lElse {
                     $v.bestCut = v.bestCut \cup s.bestCut$ 
                }
            }
        }
    }
    \If {$atLeastOne$} {
        $v.bestPlan=v.bestPlan \cup \{ v \}$ \;
        \If {$v.materialized$} {
            \lIf {$v.bestCut == null\ or\ benefit(v.bestCut, v) < benefit(\{v\}, v)$} {
                $v.bestCut = \{v\}$ 
            }
        }
        \If {$benefit(v.bestCut, v) \geq 0$} { \label{line:alproof}
            $v.bestPlan = REWRITE(v.bestPlan, v.bestCut)$ \;
            $v.bestCut = null$ \;
        }
    }
}

\end{algorithm}
\setlength{\textfloatsep}{0ex}

\subsubsection{Handling adaptive optimization}
\label{parcur:sec:integration-roulette}

We implement \parcur\ by extending \roulette, which uses adaptive sharing-aware optimization. \roulette\ splits batch execution into episodes, which last for the duration of processing one small base table vector each, and potentially uses a different global plan in each episode. \roulette\ learns the cost of different subplans across episodes and eventually converges into an efficient global plan.

The episode-oriented design conflicts with two-phase optimization: the reuse phase chooses subexpressions to reuse  based on the baseline plan at partition granularity, whereas \parcur\ switches between multiple baseline plans during a partition's execution. We reconcile the two using the concept of \textit{mini-partitions}. Mini-partitions are horizontal splits of 1-st level partitions and are internally organized using 2-nd level blocks. \parcur\ splits the base table's 1-st level partitions into fixed-size mini-partitions and then splits materializations such that tuples derived from the same base table mini-partition are clustered together.

\parcur\ makes reuse decisions at the mini-partition granularity. When accessing a mini-partition for the first time, \parcur\ chooses a baseline plan and makes two decisions: i) it decides whether the baseline plan is stable, i.e.,  it checks whether it is still learning the cost of the used subplans by tracking changes in cost estimates, and ii) if the plan is stable, it uses the reuse phase to choose materializations to use. Then, until the mini-partition is finished, \parcur\ retains the reuse decisions and optimizes the downstream computations for each reused subexpression independently.

\begin{figure}[t]
 \centering
	\includegraphics[width=0.7\hsize]{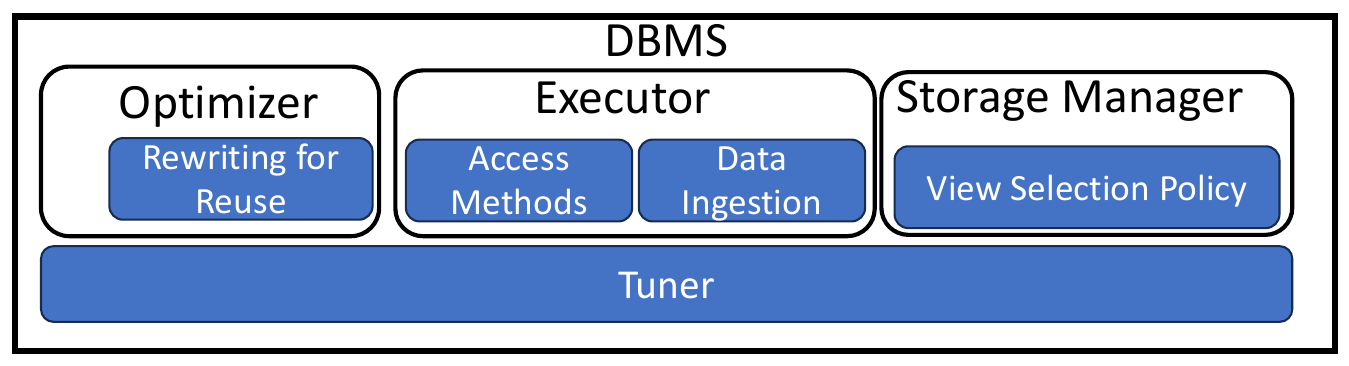}
 \vspace{-3mm}
    \caption{DBMS components that \parcur\ modifies}
	\label{parcur:fig:dbms}
    \vspace{-1mm}
\end{figure}

\vspace{-0.1in}
\section{Implementation}
\label{parcur:sec:impl}

We implement \parcur\ on RouLette \cite{sioulas2021scalable}. We modify the \textit{policy} and \textit{ingestion} components, we introduce a materialization operator, and add the tuner's utilities. In general, as shown in Fig \ref{parcur:fig:dbms}, \parcur\ interacts with i) the optimizer to receive a plan and rewrite it, ii) with the executor in order to achieve partition-at-a-time execution and use the available access methods and iii) with the storage manager so that it selects the right views for materialization.

\textbf{Tuning the Cost Model}. The presented techniques rely on \roulette's cost models. We use the same constant factors and also introduce the new constant $c_f$ (Sec \ref{parcur:sec:injecting-reuse}) which is set to $c_f=139.45$ after using regression to fit filtering cost-estimates.

\textbf{Tuning Partitioning}. To tune the parameters for the two levels of partitioning, we use the workload of Figure \ref{parcur:fig:sensitivity}a and we find the minimum values for mini-partition size and block size, and the maximum sampling rate, such that overhead is less than $10\%$ compared to the optimal value. We set the minimum size of mini-partitions to $2^{16}$ to maintain low overhead for the reuse phase and set $PS_{min}=2^{16}$ as 1st-level partitions contain at least one mini-partition. To keep the overhead for data and filter-skipping low, we select the size of mini-partitions to be greater or equal to $2^{16}$ and at least large enough that the blocks contain at least $256$ tuples each on average. Finally, to avoid significant overhead for tracking historical accesses, we set the sampling rate to $1\%$.

\textbf{Limitations}. To combine reuse with adaptive optimization, \parcur's implementation over \roulette\ aligns the mini-partitions of materializations with the mini-partitions of a base table. For this reason, tuning revolves around one main table that defines the partitioning schemes and the materializations. Our implementation is applicable to common workloads such as queries on star and snowflake schemas. Also, our prototype optimizes execution but tuning is single-threaded. Deciding on the frequency of tuning, the amount of resources, and how sync with execution should happen are well-known problems but orthogonal to ours.

\textbf{In-memory vs disk-based}. While \parcur\ relies on an in-memory system, the performance trends are not expected to change if we transition to a disk-based implementation. With modern SSDs and large query batches, data access would still be fast, whereas shared filtering of materialized results would continue to be expensive. Thus, we expect different speedup due to different tradeoffs, but the main insights would still be valid.
\vspace{-0.15in}
\section{Experimental Evaluation}
\label{parcur:sec:evaluation}

The experiments evaluate \parcur\ and show how materialization and reuse enable it to significantly outperform pure work sharing and achieve lower batch response times. Specifically, they demonstrate the following:

\noindent i) Filtering costs when accessing materializations can deteriorate the performance of work sharing, and thus building access methods for materializations is necessary.

\noindent ii) Query-at-a-time materialization policies make suboptimal materialization decisions. Cut selection improves budget utilization by prioritizing materialization with higher marginal benefits.

\noindent iii) Homogeneity-based partitioning reduces the required budget for workloads with selective and correlated patterns.

\noindent iv) Even though filters change, the reuse phase reduces work-sharing's response time when possible and falls back to vanilla work sharing otherwise.

\noindent v) Using partial reuse, the response time is proportional to the required computation and performance degrades gracefully.

\noindent vi) End-to-end, \parcur\ reduces the response time for the full SSBM and TPC-H by $6.4\times$ and $2\times$, respectively.

\textbf{Hardware}. All experiments took place on a single server that features an Intel(R) Xeon(R) Gold 5118 CPU @ 2.30GHz with $2$ sockets, $12(\times2)$ threads per socket, $376$GB of DRAM, $32$KB L1 cache, $1$MB L2 cache and $16$MB L3 cache. All experiments took place in memory, in a single NUMA node, and use $12$ threads.

\textbf{Data \& Workload}. We run both macro- and micro-benchmarks. First, we perform a sensitivity analysis. We evaluate \parcur\ by varying different workload properties: i) the number of filtering attributes, ii) the selectivity of predicates, iii) the number of joins and the overlap between queries, iv) the available budget, and v) the workload shift in filter attributes and predicate correlations. To control the experiment variables, we generate synthetic data in a star schema as well as appropriate queries. We use a fact table of $100M$ rows and $27$ columns ($24$ are foreign keys), $8$ dimensions with $10k$ rows and $9$ columns each, and $16$ dimensions with $10k$ rows and $2$ columns each. All columns are $4$-byte integers. We describe the queries in the presentation of each micro-benchmark.

Next, we show that \parcur\ accelerates the queries of the widely used SSBM \cite{ssb} and TPC-H benchmarks. We use SF10 for both, which is the largest data size for which the optimal materialization fits in memory. We randomize the order of tuples for both datasets.

\textbf{Methodology}. The experiments measure batch response time, which is the end-to-end time for processing the full batch. All measurements are the average of $10$ runs.

\vspace{-3mm}

\subsection{Impact of Reuse in Global Plans}
\label{parcur:sec:reuse-exp}

\begin{figure*}[t]
\begin{tabular}{@{}l@{}l}
\begin{minipage}{0.33\hsize}
    \centering\includegraphics[width=\hsize,page=14]{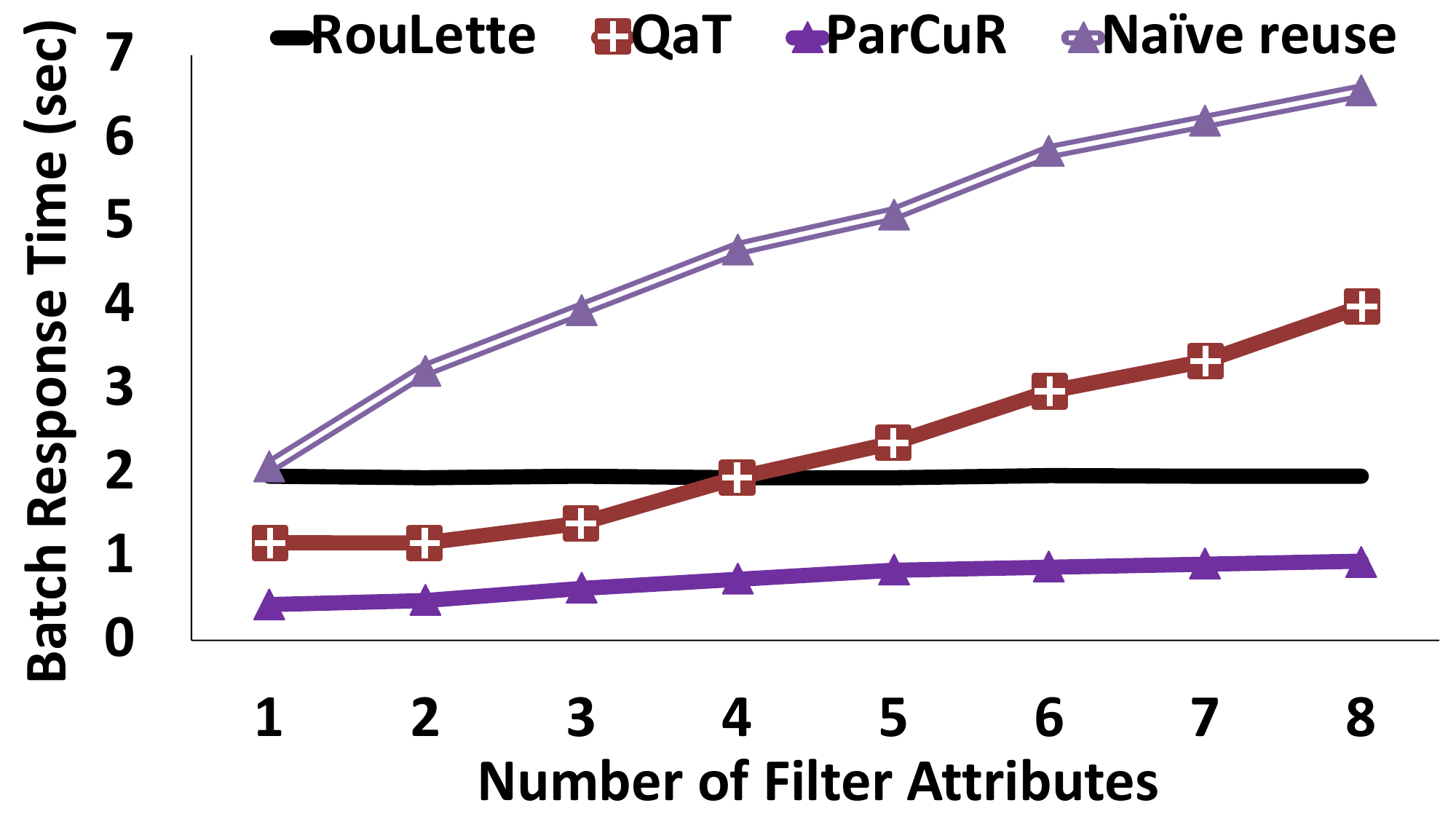}	\\
    \vspace{-1mm}
        \centering {\footnotesize(a)}
        \end{minipage}
        &
    \begin{minipage}{0.33\hsize}
\centering\includegraphics[width=\hsize,page=13]{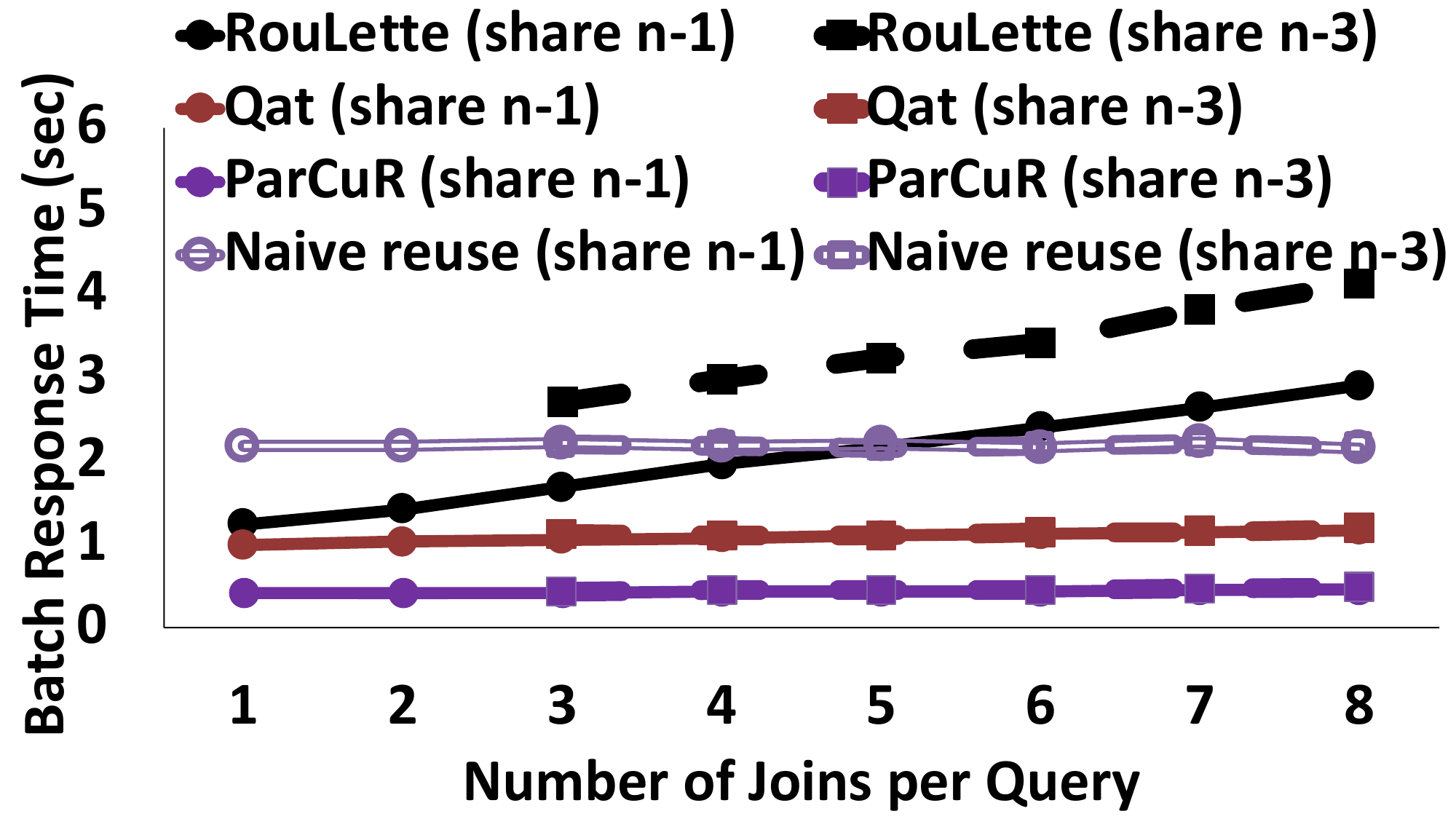}	\\
\vspace{-1mm}
    \centering {\footnotesize(b)}
    \end{minipage}
        \begin{minipage}{0.33\hsize}
\centering\includegraphics[width=\hsize,page=15]{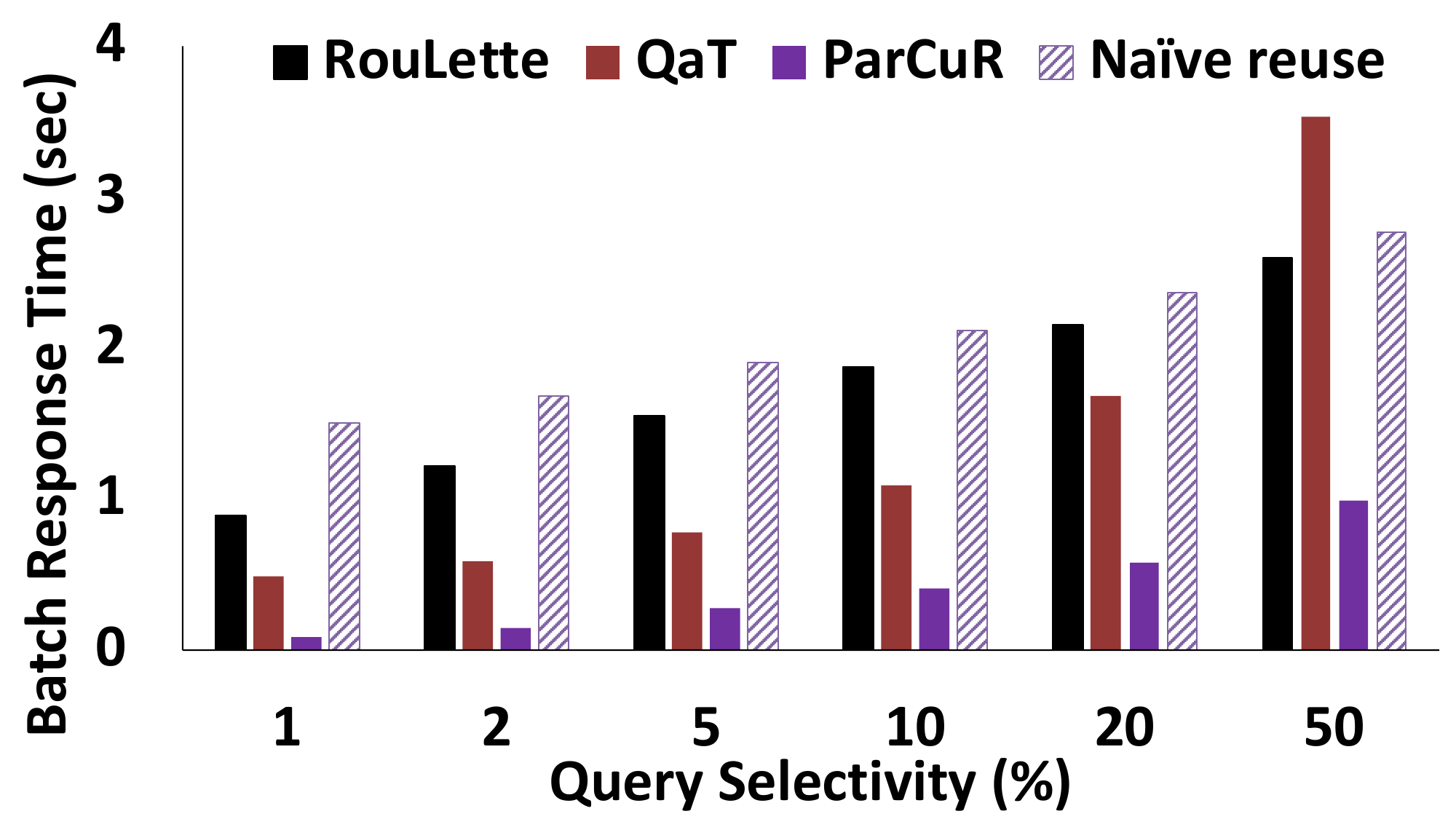}	\\
\vspace{-1mm}
    \centering {\footnotesize(c)}
    \end{minipage}
\end{tabular}
\vspace{-2mm}
\caption{\label{parcur:fig:sensitivity} Impact of reuse based on workload parameters in a) filters, b) joins, and c) selectivity.}
\vspace{-2mm}
\end{figure*}



We evaluate the benefit of reuse to shared execution's response time. We assume that the tuner's workload is the same as the runtime workload and that the materializations that minimize response time are available (i.e., the top-level joins). Sections \ref{parcur:sec:caching-exp} and \ref{parcur:sec:work-shift-exp} lift the two assumptions. We compare \parcur\ against \roulette, naive reuse, which eagerly injects materializations and has no access methods, and QaT execution using \roulette, which is on par with QaT performance of state-of-the-art in-memory databases.

\noindent \textbf{Filter processing}. We examine the impact of filters and the need for building and using access methods for materializations. We use $64$ queries generated from $4$ different templates. The templates have $4$ dimension joins each, and all templates share $3$ dimension joins. The queries have $10\%$ selectivity and filter on the non-shared dimension. We vary the number of filter attributes (which is equal to the number of shared filter operators) from $1$ to $8$.

Figure \ref{parcur:fig:sensitivity}a shows that access methods are necessary for accelerating work sharing. When using access methods, \parcur's response time is $2.07$-$4.57\times$ lower than \roulette's, as it eliminates join processing. \roulette\ is almost unaffected by increasing filter operators, as it processes filters on the dimension. \parcur\ and QaT are affected because they require more 2-nd level partitions and hence both more zone-map operations as well as larger mini-partitions, and thus longer time until \parcur\ decides that the plan is stable. However, this effect just reduces \parcur's benefit over \roulette. By contrast, the performance of naive reuse deteriorates drastically: it computes filters over the materialization and thus their processing time is amplified. The response time is increased with the number of filters and is up to $3.34\times$ than \roulette's.

\noindent \textbf{Takeaway}: Reuse drastically improves performance only if filtering cost is low. Building appropriate access methods is necessary for injecting materializations into global plans.



\noindent \textbf{Number of joins}. We examine the impact of reuse in queries with different join costs. We use two variants of the previous workload, one where all templates share all but one join (share n-1) and another where they share all but three joins (share n-3). We vary the total number of joins per query. All queries use $1$ dimension filter.

Figure \ref{parcur:fig:sensitivity}b shows larger benefits for global plans with more joins. Reuse-based approaches are insensitive to the number of joins, whereas \roulette's response time is increased. \parcur\ achieves maximum speedup of $6.33$ for share n-1 and $8.60$ for share n-3. Also, there is a cross-point in naive reuse, where processing filters becomes preferable to large joins.

\noindent \textbf{Takeaway}: The benefit from reuse is proportional to the eliminated computation. Hence, the speedup is higher when eliminated computation is significant, such as in join-heavy queries.

\noindent \textbf{Selectivity}. We examine the impact of reuse for queries with different selectivity. We use the same workload as in the first experiment, use one filter attribute, and vary the selectivity ($1\%$, $2\%$, $5\%$, $10\%$, $20\%$, $50\%$). The experiment models the impact of downstream processing.

Figure \ref{parcur:fig:sensitivity}c shows larger benefits when each query's selectivity is low. As aggregations  are not affected by reuse, they close the gap between approaches for larger selectivity when they are expensive. Also, it is noteworthy that when aggregations are heavy enough, QaT is more expensive than \roulette\ due to concurrency.

\noindent \textbf{Takeaway}: Reuse has a higher benefit when it eliminates the most expensive part of the global plan. Low selectivity results in low-cost final aggregation, and thus the relative benefit is more pronounced. Nevertheless, reuse is the best approach across all selectivities.

\vspace{-3mm}

\subsection{Sharing-aware Materialization Policy}
\label{parcur:sec:caching-exp}

\begin{figure}[t]
\begin{tabular}{@{}l@{}l}
\begin{minipage}{0.5\hsize}
    \includegraphics[width=\hsize,page=19]{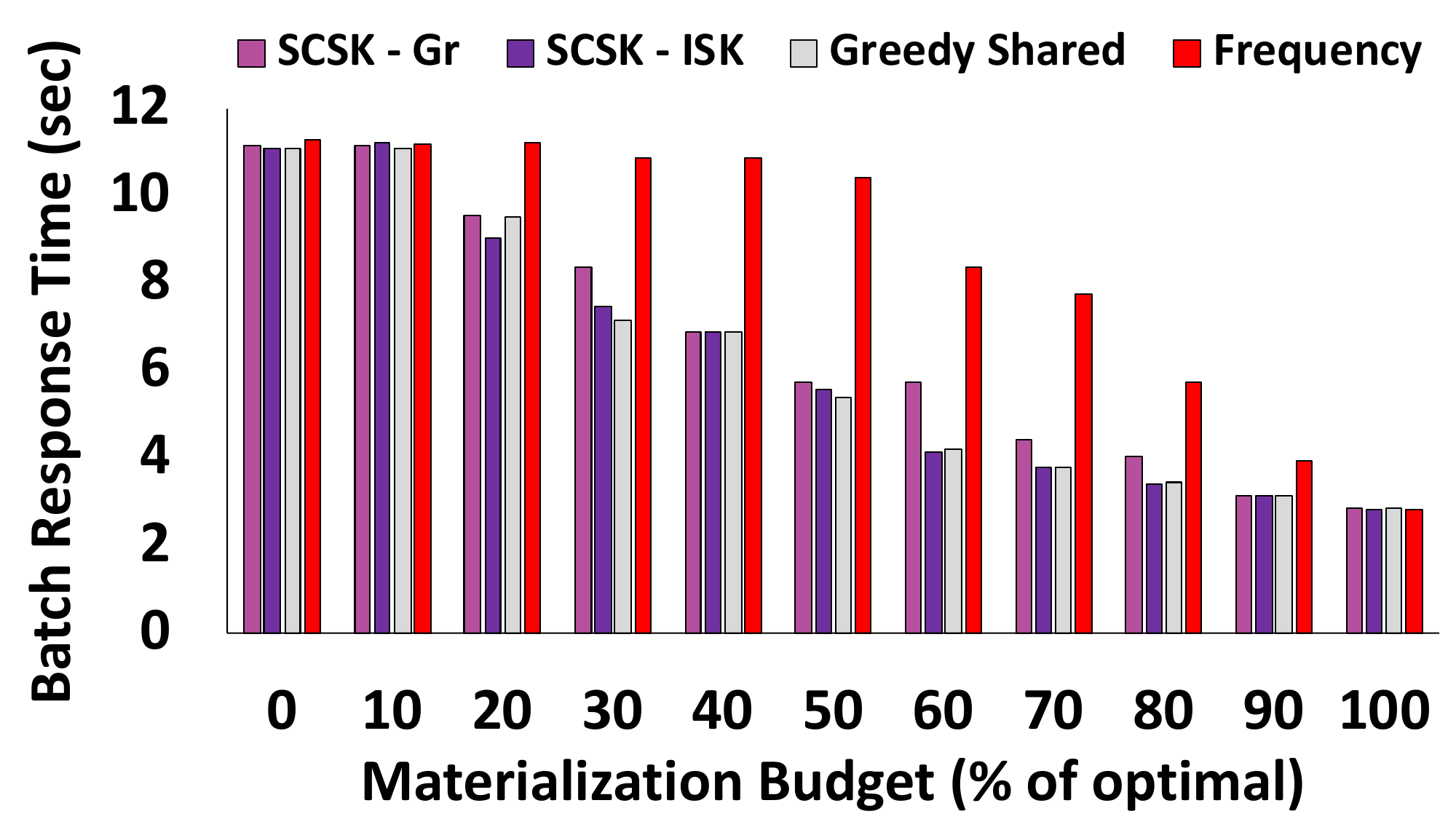}	\\
    \vspace{-1mm}
        \centering {\footnotesize(a)}
        \end{minipage}
        &
    \begin{minipage}{0.5\hsize}
\includegraphics[width=\hsize,page=20]{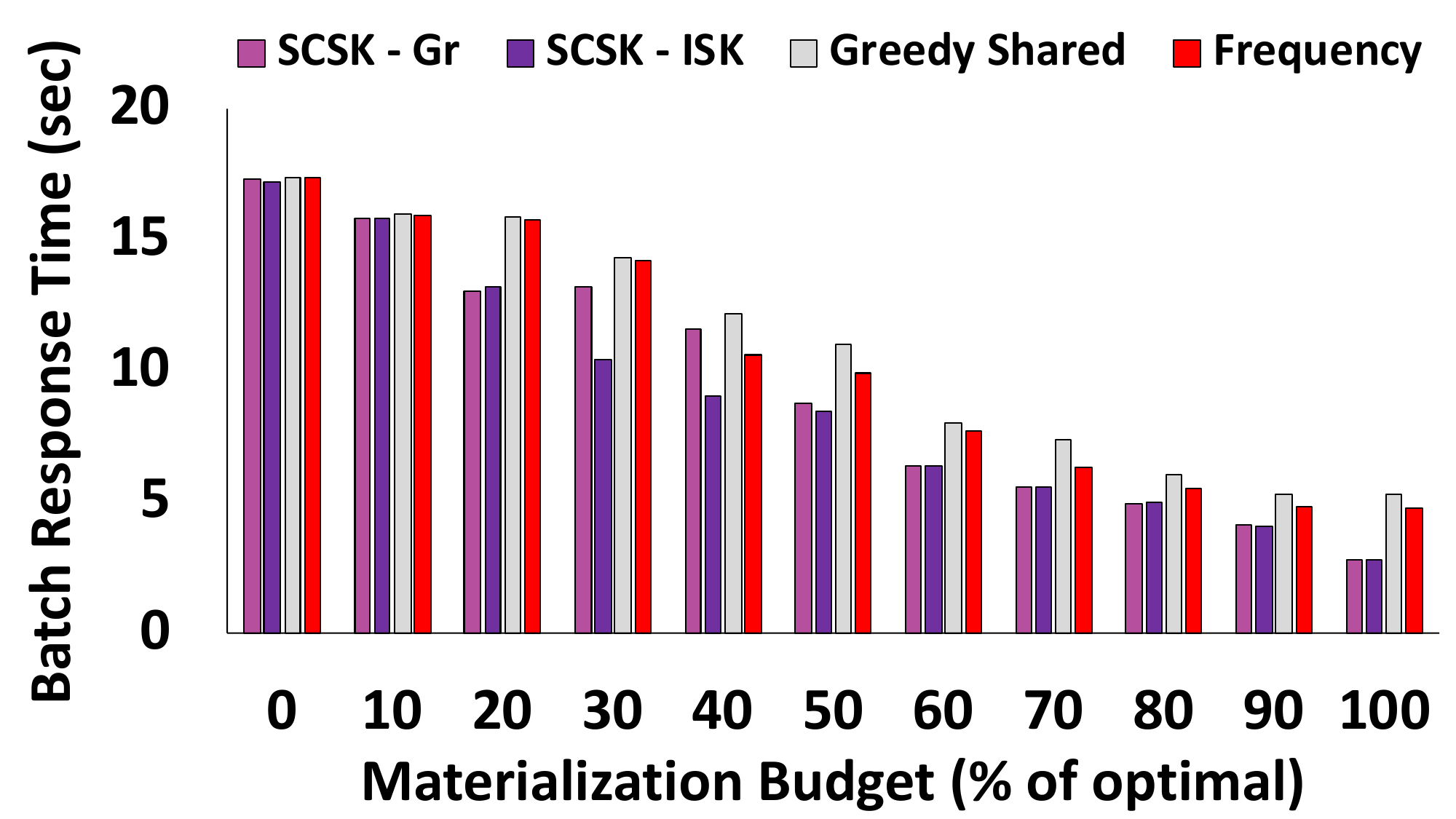}	\\
\vspace{-1mm}
    \centering {\footnotesize(b)}
    \end{minipage}
\end{tabular}
\vspace{-2mm}
\caption{\label{parcur:fig:budget-ab} Impact of budget for workloads A and B.}
\vspace{-2mm}
\end{figure}


\begin{figure}[t]
\begin{tabular}{@{}l@{}l}
\begin{minipage}{0.5\hsize}
    \includegraphics[width=\hsize,page=21]{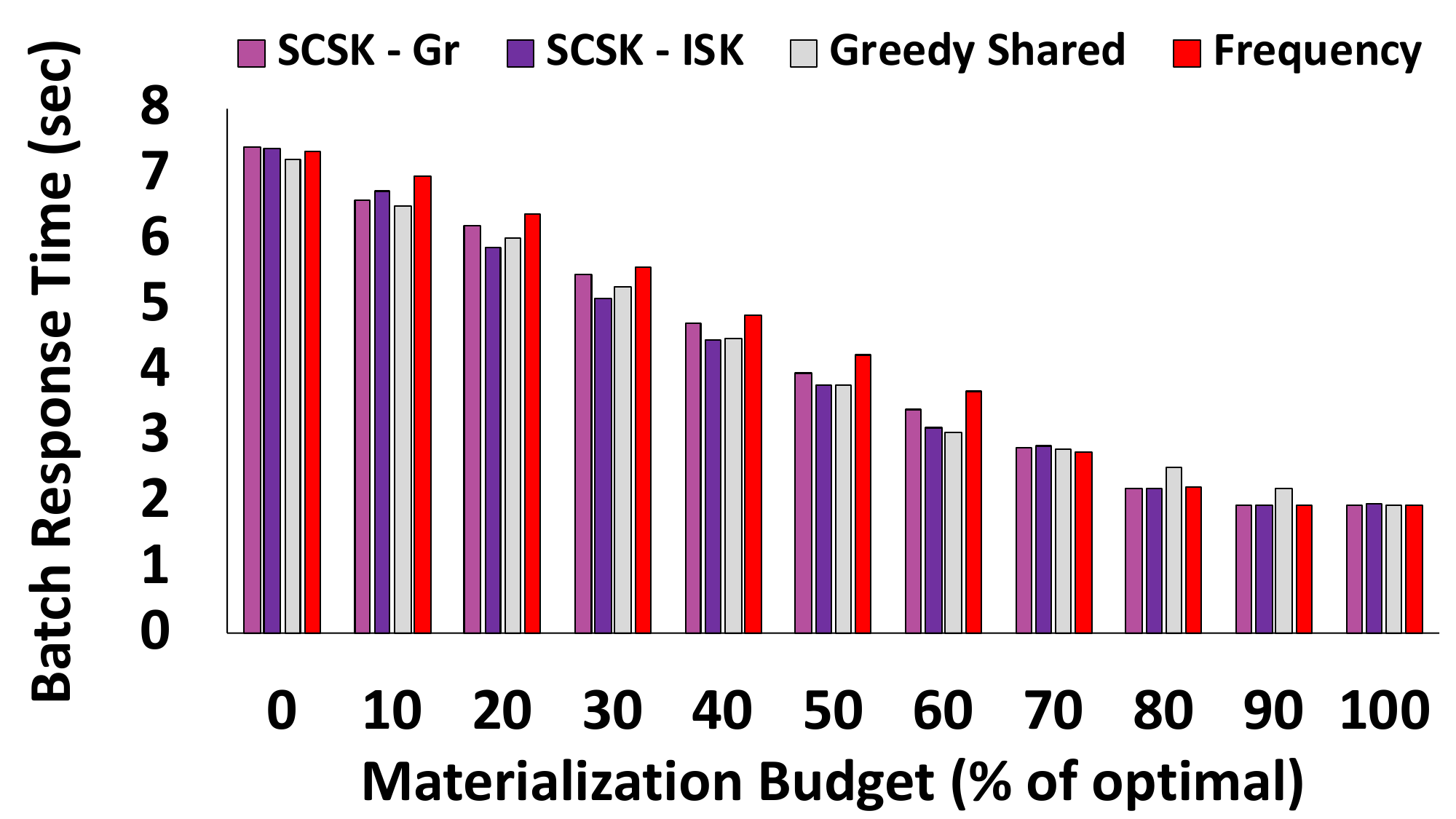}	\\
    \vspace{-1mm}
        \centering {\footnotesize(a)}
        \end{minipage}
        &
    \begin{minipage}{0.5\hsize}
\includegraphics[width=\hsize,page=22]{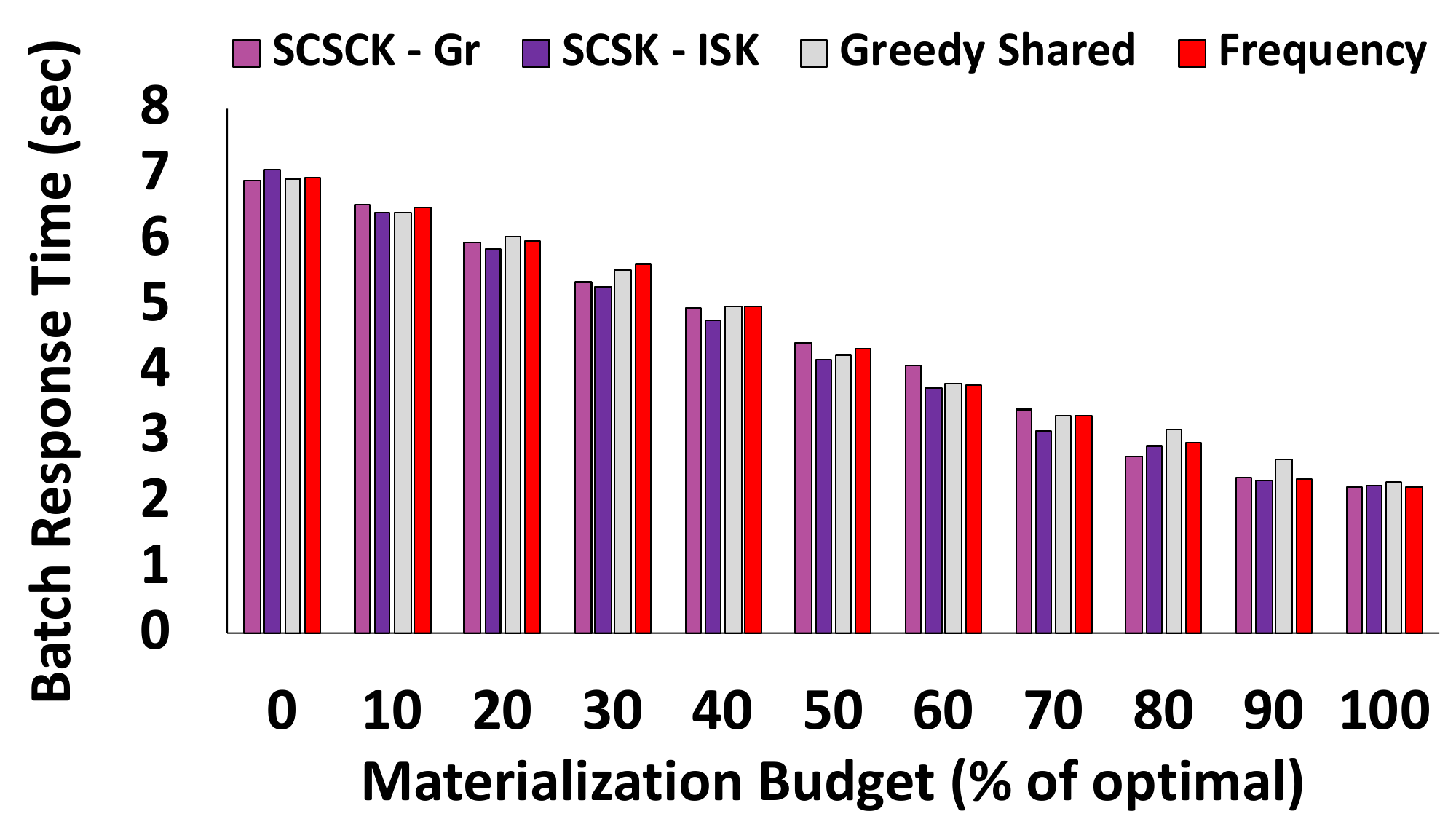}	\\
\vspace{-1mm}
    \centering {\footnotesize(b)}
    \end{minipage}
\end{tabular}
\vspace{-2mm}
\caption{\label{parcur:fig:budget-abp} Impact of budget for workloads B-P1 and B-P2.}
\end{figure}


We demonstrate cut selection solutions outperform sharing-oblivious and simple sharing-aware policies. We compare four different algorithms: a) \textbf{SCSK-Gr} solves cut selection using Gr, b) \textbf{SCSK-ISK} solves cut selection using ISK, c) \textbf{Greedy Shared} solves a submodular knapsack problem for individual materializations, and d) \textbf{Frequency} solves the submodular knapsack problem where benefits are weighted by frequency, which is commonly used for query-at-a-time materialization. The evaluation uses four different workloads with $512$ queries with $10\%$ selectivity each. The queries use filters in a column with domain $[0, 100)$. At the end of each workload, we mention the amount of DRAM it requires to minimize response time.

- \emph{Workload A} shows the impact of frequency. It uses 8 query templates ($t_1, \dots, t_8$). $t_1, \dots, t_4$ have $1$ join each, whereas $t_5, \dots, t_8$ have $4$. Template $t_i$ shares its join with $t_{i+4}$. The workload contains $112$ queries from each of  $t_1, \dots, t_4$ and $16$ queries from $t_5, \dots, t_8$. Requires at least $40$GB.

- \emph{Workload B}: it uses 8 query templates ($t_5, \dots, t_{12}$). $t_9, \dots, t_{12}$ also have $4$ joins each. Template $t_i$ shares $2$ joins with template $t_{i+4}$. The workload contains $64$ queries from each of the templates. The workload shows the impact of synergy. Requires at least $32$GB.

- \emph{Workload B-P1}: it uses workload B's templates. However, the filters for $t_5$ and $t_6$ are subranges of $[0, 40)$, for $t_9$ and $t_{11}$ subranges of $[20, 60)$, for $t_7$ and $t_{8}$ subranges of $[40, 80)$, and for $t_{10}$ and $t_{12}$ subranges of $[60, 100)$. Requires at least $14.9$GB.

- \emph{Workload B-P2}: Similar to workload B-P1, but uses 2-D ranges. The filters for $t_5$ and $t_6$ are subranges of $[0, 66)\times [0, 66)$, for $t_9$ and $t_{11}$ subranges of $[0, 66)\times [34, 100)$, for $t_7$ and $t_{8}$ subranges of $[34, 100)\times [0, 66)$, and for $t_{10}$ and $t_{12}$ subranges of $[34, 100)\times [34, 100)$. Requires at least $12.9$GB.

In each experiment, we vary the storage budget to the minimum that can minimize response time. We present the used budget normalized by the one that minimizes response time (i.e., 100\%). Vanilla \roulette\ corresponds to 0\% budget for all policies.

\noindent \textbf{Sharing-awareness}: Figure \ref{parcur:fig:budget-ab}a shows that sharing-aware policies outperform Frequency in workload A because they factor out the frequency of occurrence for subqueries, and decide based on shared costs. Frequency results in up to $2.03\times$ higher response time for the same budget because it prioritizes templates $t_1, \dots, t_4$. 

\noindent \textbf{Synergy-awareness}: Figure \ref{parcur:fig:budget-ab}b shows that exploiting the synergy between materializations that compose cuts in workload B improves the effectiveness of materializations. Both Greedy Shared and Frequency preferentially materialize the shared subqueries because they miss the synergy between the larger cuts. Thus, they both waste budget on materializing subexpressions that are later covered by the larger cuts, and consequently, $100$\% is not sufficient for minimizing response times. At $100\%$, they are slower by $1.87\times$ and $1.68\times$, respectively.

\noindent \textbf{Partition-awareness}: For both workload B-P1 and B-P2, partitioning reduces the required budget for minimizing response times by $2.4\times$ and $2.5\times$ accordingly. Figure \ref{parcur:fig:budget-abp} shows that all algorithms achieve comparable performance because partitioning simplifies the global plans for each partition. The simplification mitigates the effect of synergy and frequency, and thus all algorithms find comparable solutions.

\noindent \textbf{Gr vs ISK}: Across all experiments,  ISK performs better than Gr as it enumerates more materializations and normalizes marginal benefit by the required budget. By contrast, Gr suffers from suboptimal solutions when it uses up the budget on few materializations. Still, ISK requires significant processing time to run, e.g., $217$sec in workload B-P1, and thus Gr is preferable for real-time analysis as it takes up to $4$msec in all experiments.

\noindent \textbf{Takeaway}: Both sharing-awareness and partitioning improve budget utilization. Incorporating both shared costs and synergy permits spending the budget for materializing only the subexpressions that actually reduce response times. Furthermore, partitioning enables materializing results just for the data ranges where they are needed and thus reduces budget requirements.

\subsection{Effect of workload shift in reuse}
\label{parcur:sec:work-shift-exp}

\begin{figure}[t]
\begin{tabular}{@{}l@{}l}
\begin{minipage}{0.5\hsize}
    \centering\includegraphics[width=\hsize,page=16]{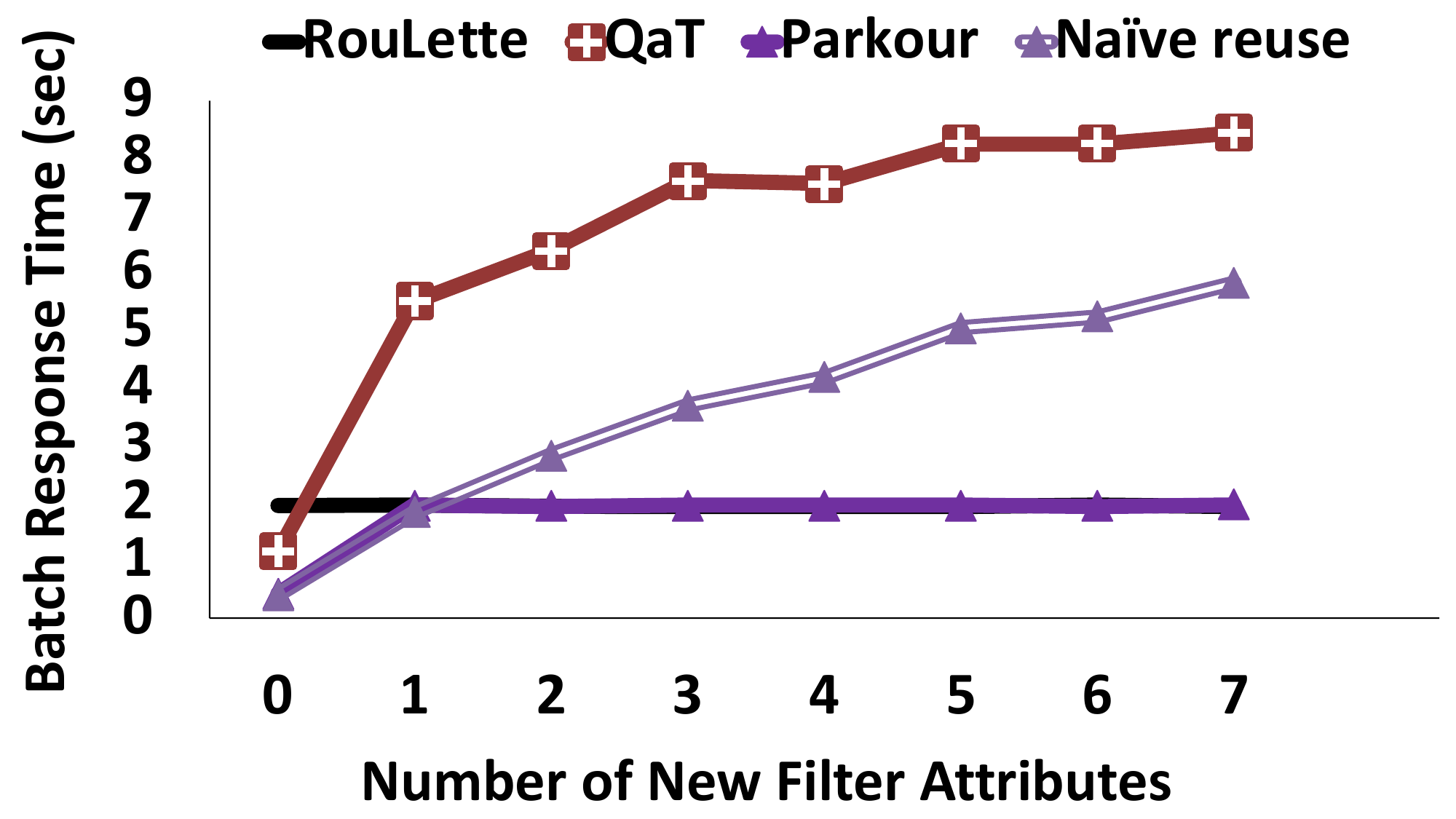}	\\
    \vspace{-1mm}
        \centering {\footnotesize(a)}
        \end{minipage}
        &
        \begin{minipage}{0.5\hsize}
\centering\includegraphics[width=\hsize,page=6]{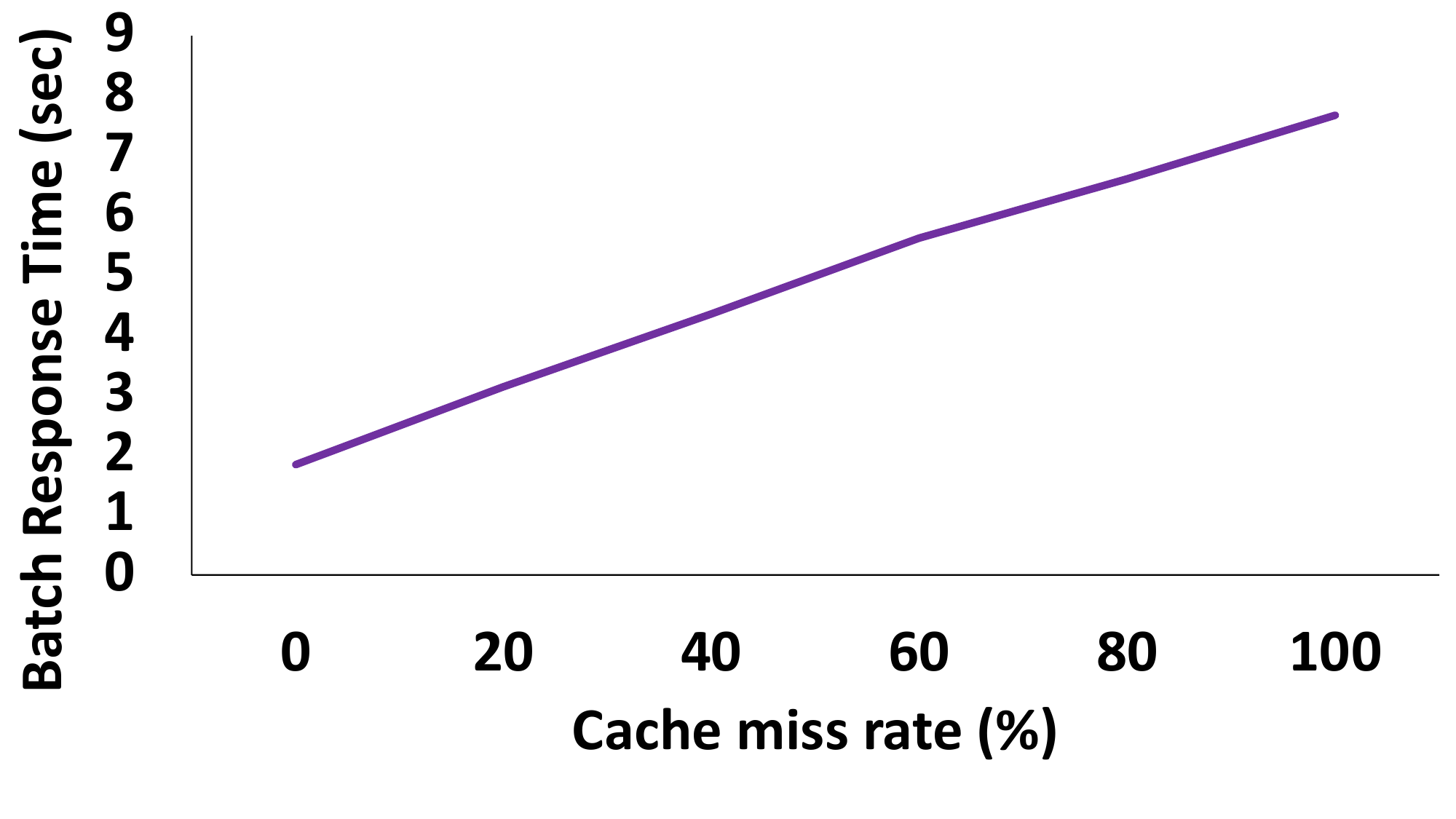}	\\
\vspace{-1mm}
    \centering {\footnotesize(b)}
    \end{minipage}
\end{tabular}
\vspace{-2mm}
\caption{\label{parcur:fig:work-shift} Impact of workload shift in a) filtering attributes, and b) query patterns' predicates.}
\end{figure}




We evaluate \parcur\ under workload shift. We materialize subexpressions that minimize the response time for the original workload. The experiments shift workload across two axes: a) by adding new filtering attributes, and b) by shifting query pattern predicates. We compare \parcur\ against \roulette{}, naive reuse (for which we enable access methods), and QaT.

\noindent \textbf{Filtering attributes}. Figure \ref{parcur:fig:work-shift}a shows that the reuse phase judiciously chooses between reuse and recomputation based on filtering costs. The experiment uses the same workload as Figure \ref{parcur:fig:sensitivity}a. We assume that the original workload is the batch with one filtering attribute, hence we only build an access method for that attribute. Naive reuse improves response time when there is no shift and deteriorates performance otherwise. QaT's performance depends on the percentage of queries that use the materializations. Finally, \parcur\ improves performance when there is no shift and achieves the same performance as work sharing when reuse is detrimental.


\noindent \textbf{Query patterns' predicates}. Figure \ref{parcur:fig:work-shift}b shows that partial reuse enables response times to degrade gracefully under workload shift. The experiment uses workload B-P1 to build materializations. The shifted workload slides the ranges for the filters of each template; the slide controls the percentage of the shifted workload's input that cannot reuse materializations and is processed from base data (miss rate). \parcur's response time is increased proportionally to the miss rate. Thus, when partitioning captures query patterns and isolates misses, partial reuse improves performance against all-or-nothing approaches that fall back to full processing (same performance as 100\% miss rate).

\noindent \textbf{Takeaway}: The reuse phase, as well as partitioned execution, enable \parcur\ to benefit from materializations despite workload shifts. \parcur\ exploits materializations for the partitions where they are available and beneficial to reducing the global plan's cost.

\subsection{Macro-benchmarks}


\begin{figure}[t]
\begin{tabular}{@{}l@{}l}
\begin{minipage}{0.5\hsize}
    \includegraphics[width=\hsize,page=23]{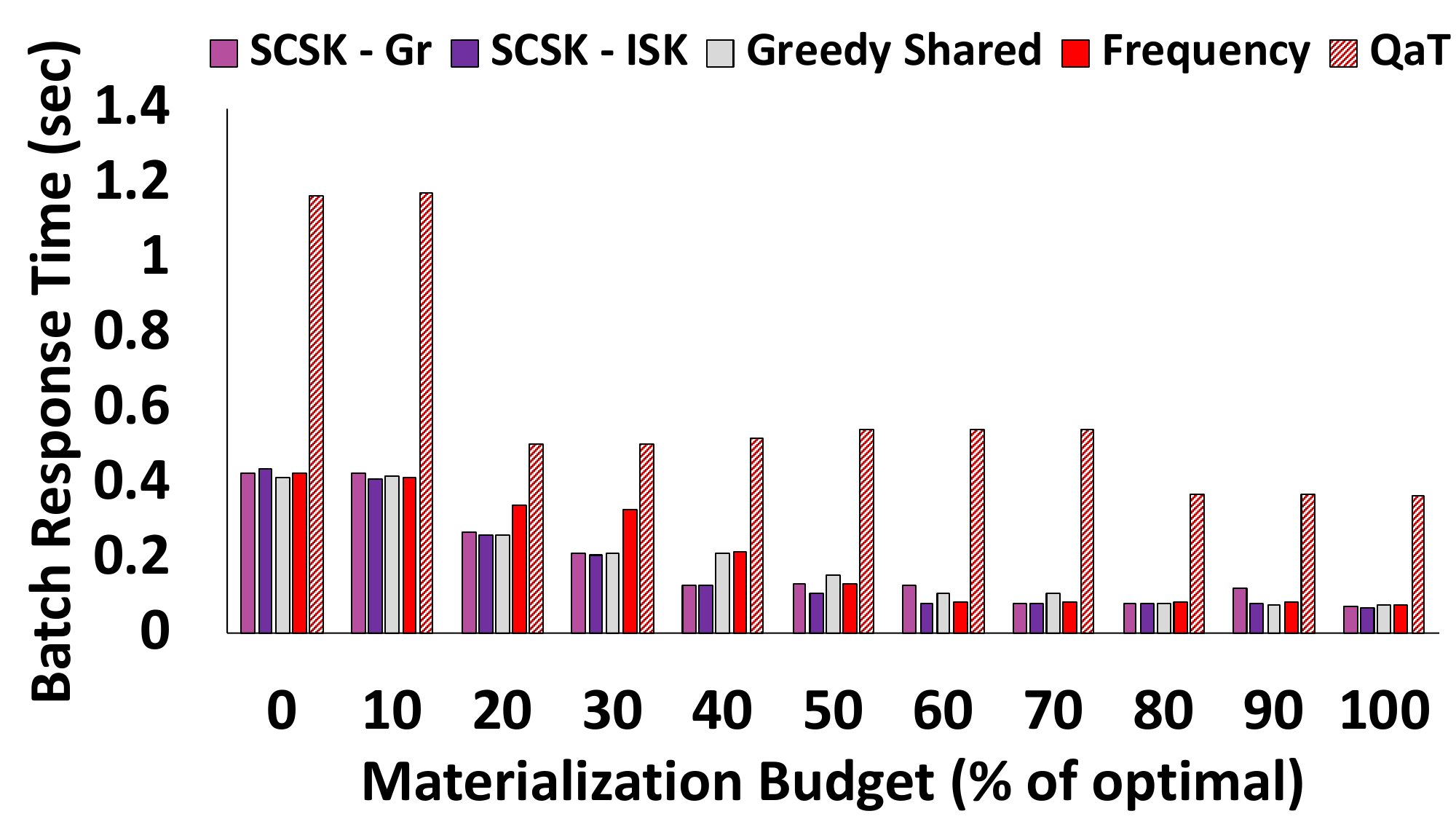}	\\
    \vspace{-1mm}
        \centering {\footnotesize(a)}
        \end{minipage}
        &
    \begin{minipage}{0.5\hsize}
\includegraphics[width=\hsize,page=24]{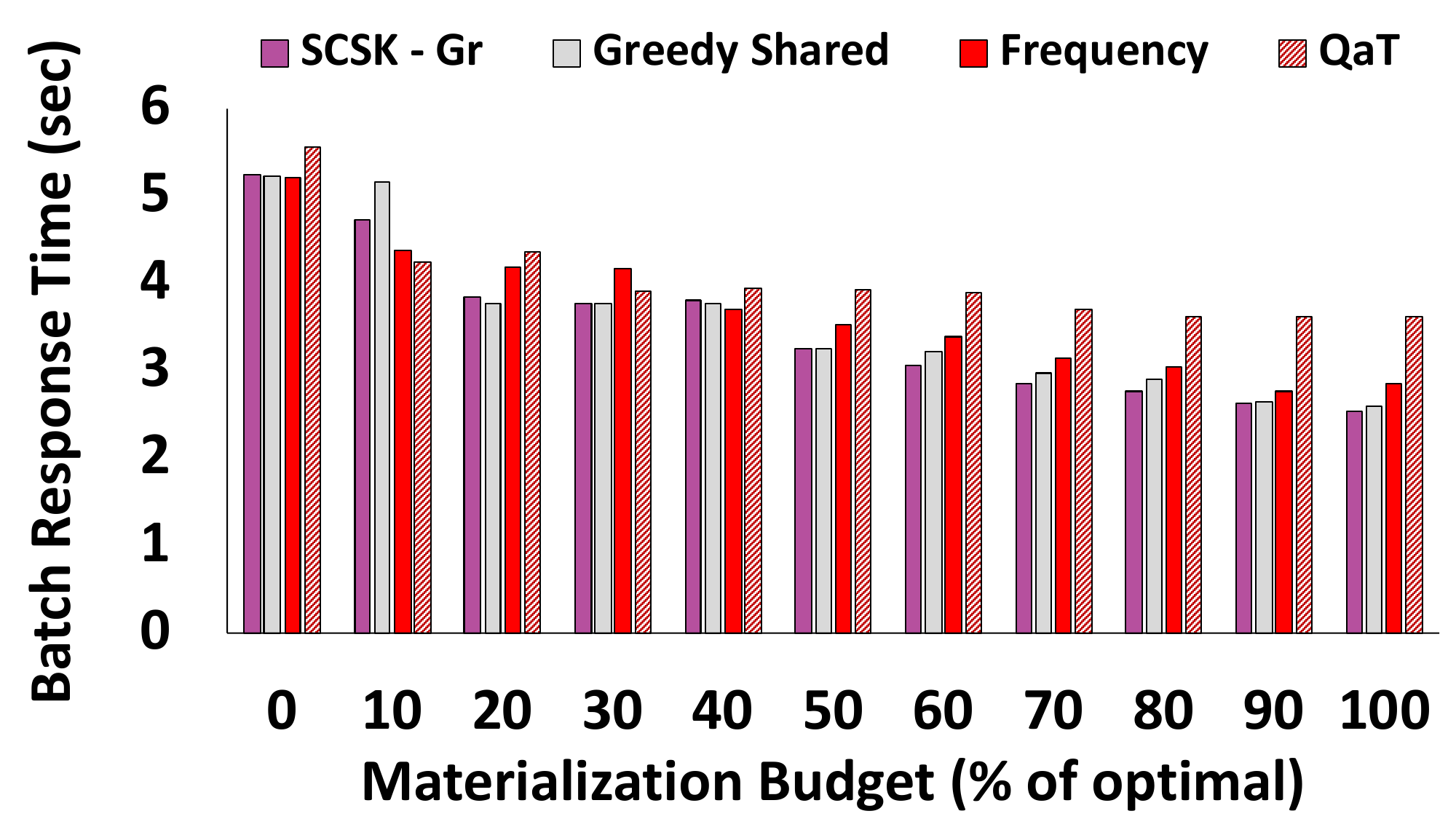}	\\
\vspace{-1mm}
    \centering {\footnotesize(b)}
    \end{minipage}
\end{tabular}
\vspace{-4mm}
\caption{\label{parcur:fig:macro} Macro-benchmarks: a) SSB b) TPC-H.}
\end{figure}

We evaluate \parcur\ using the SSBM and TPC-H benchmarks, which contain 13 and 22 queries respectively. For each benchmark, we compare the four materialization algorithms  and vary the storage budgets. We omit ISK for TPC-H, because it takes a very long to choose a materialization.

\noindent \textbf{SSBM}: Figure \ref{parcur:fig:macro}a shows that \parcur\ achieves a maximum speedup of $6.4$ over \roulette, which corresponds to 0\% budget, and $5.4$ over QaT, and requires around $1$GB for the optimal materialization. The speedup is high because queries are mostly selective, and thus aggregations make up a small percentage of processing time; the vast majority is filters and joins. An interesting observation is that even a small budget, at 20\%, brings about a 37\% decrease in response time because bottom joins are significantly more expensive, whereas upper joins are more selective and less time-consuming.

\noindent \textbf{TPC-H}: Figure \ref{parcur:fig:macro}b shows that \parcur\ achieves a maximum speedup of $2\times$ over \roulette\ and $1.37\times$ over QaT, and requires $69$GB for the optimal materialization. The speedup is lower compared to SSB for two reasons: i) TPC-H also contains less selective queries with heavier aggregations. When using $100$\% budget, aggregation takes up around $40$\% of the processing time. ii) TPC-H contains LIKE predicates that filter skipping cannot eliminate using zone-maps. Even so, despite the shortcomings in our implementation, \parcur\ eliminates significant join costs.

\noindent \textbf{Discussion}: For the two benchmarks, \parcur\ requires large materializations because, our homogeneity-based partitioning does not exploit filters on dimensions. This limitation can be addressed by: i) partitioning using the denormalized table \cite{yang2020qd}, or ii) using data-induced predicates on the fact table's foreign keys \cite{kandula2019}. Both techniques are 
straightforward to integrate with \parcur.

Another limitation is that \parcur\ cannot eliminate predicates such as LIKE, multi-attribute expressions, or UDFs using zone-maps. To eliminate such predicates, partitions require additional metadata. Sun et al. \cite{sun2014fine} handle such predicates by maintaining a feature vector that encodes whether complex predicates  are satisfied.

\vspace{-0.1in}
\section{Related Work}
\label{parcur:sec:related}

We compare \parcur\ with related work in (i) sharing, (ii) materialization, and (iii) partitioning.

\textbf{Work sharing:} Work sharing exploits overlapping work between queries in order to reduce the total cost of processing. Despite using diverse execution models and optimization strategies, recent work-sharing databases use global plans \cite{harizopoulos2005qpipe, arumugam2010datapath, karimov, giannikis2012shareddb, mqjoin, sioulas2021scalable} and the Data-Query model \cite{cacq, precision-sharing, candea2009scalable, arumugam2010datapath, giannikis2012shareddb, mqjoin, sioulas2021scalable, astream, karimov}. Existing work-sharing databases do not support reuse; they always recompute global plans from scratch. \parcur\ is compatible with such databases, and therefore this work's insights are valuable for reducing their response time for recurring workloads.

\textbf{Reuse:} Reuse occurs in different forms, such as semantic caching \cite{shimcaching, deshpandecaching, darcaching, melvincaching}, recycling \cite{ivanovarecycle, nagelrecycle, tanrecycle, perezviews}, view selection \cite{roussopoulosviews, kalnis2002view, zhang2001evolutionary, mami2011modeling}, and subexpression selection \cite{zhou2007efficient, jindal2018selecting, jindalreuse}. \parcur\ addresses subexpression selection in the context of sharing environments. Sharing affects the data layout, the materialization policy, and the reuse policy for the selected subexpressions. This is the first work that studies the effect of work sharing on reuse. Extending semantic caching, recycling, and view selection for shared execution is a significant direction for future work.

\parcur\ also supports partial reuse. Similar approaches include chunk-based semantic caching \cite{deshpandecaching, darcaching}, partially materialized views \cite{zhou2007dynamic}, partially-stateful dataflow \cite{gjengset2018noria}, and separable operators \cite{yang2021flexpushdown}. However, in all of these approaches, the concurrent outstanding computation can deteriorate performance. \parcur\ both reuses available materializations and uses sharing to improve scalability.

\textbf{Partitioning:} In modern scan-oriented analytical systems, partitioning is an indispensable tool for accelerating selective queries using data skipping \cite{sun2014fine, sun2016skipping}. Existing partitioning strategies focus on minimizing data access. By contrast, \parcur\ chooses a partitioning scheme to maximize reuse while minimizing the space overhead for partition-granularity materialization. Doing so requires that partitioning captures both access and computation patterns.

\section{Conclusions}
\label{parcur:sec:conclusions}
To provide real-time responses for large recurring workloads, we propose \parcur, a novel paradigm that combines the reuse of materialized results with work sharing. \parcur\ addresses the performance pitfalls of incorporating materialized results into shared global plans i) by proposing a multi-level partitioning design that improves at the same time the utilization of the storage budget, partial reuse, and filtering costs, ii) by proposing a novel sharing-aware caching policy that improves materialization decisions, and iii) by enhancing the sharing-aware optimizer with a phase that performs reuse-oriented rewrites in order to minimize runtime processing. In our experiments, \parcur\ outperformed \roulette{} by $6.4\times$ and $2\times$ in the widely-used SSB and TPC-H benchmarks respectively.


\bibliographystyle{ACM-Reference-Format}
\bibliography{bibliography}

\clearpage
\end{document}